\RequirePackage{amsmath}
\documentclass[envcountsame,runningheads]{llncs}
\usepackage{multicol}
\usepackage{amsfonts,amssymb}
\usepackage{graphicx}
\usepackage{epic}
\usepackage{eepic}
\usepackage{epsfig,float,enumitem,setspace}
\usepackage{verbatim}
\usepackage{pdfsync}
\usepackage{gastex}
\usepackage{xcolor}
\usepackage{url}
\usepackage{geometry}

\renewcommand{\le}{\leqslant}
\renewcommand{\ge}{\geqslant}

\newcommand{\eps}{\varepsilon}
\newcommand{\emp}{\emptyset}

\newcommand{\Sig}{\Sigma}
\newcommand{\sig}{\sigma}
\newcommand{\noin}{\noindent}

\newcommand{\ur}{uniquely reachable}
\newcommand{\bi}{\begin{itemize}}
\newcommand{\ei}{\end{itemize}}
\newcommand{\be}{\begin{enumerate}}
\newcommand{\ee}{\end{enumerate}}
\newcommand{\bd}{\begin{description}}
\newcommand{\ed}{\end{description}}
\newcommand{\bq}{\begin{quote}}
\newcommand{\eq}{\end{quote}}

\newcommand{\ie}{\mbox{\it i.e.}}

\newcommand{\tid}{\mbox{{\bf 1}}}

\newcommand{\cD}{{\mathcal D}}

\newcommand{\cW}{{\mathcal W}}

\newcommand{\one}{{\mathbf 1}}
\newcommand{\Lra}{{\hspace{.1cm}\Leftrightarrow\hspace{.1cm}}}

\newcommand{\raL}{{\mathbin{\sim_L}}}
\newcommand{\lraL}{{\mathbin{\approx_L}}}

\def\shu{\mathbin{\mathchoice
{\rule{.3pt}{1ex}\rule{.3em}{.3pt}\rule{.3pt}{1ex}
\rule{.3em}{.3pt}\rule{.3pt}{1ex}}
{\rule{.3pt}{1ex}\rule{.3em}{.3pt}\rule{.3pt}{1ex}
\rule{.3em}{.3pt}\rule{.3pt}{1ex}}
{\rule{.2pt}{.7ex}\rule{.2em}{.2pt}\rule{.2pt}{.7ex}
\rule{.2em}{.2pt}\rule{.2pt}{.7ex}}
{\rule{.3pt}{1ex}\rule{.3em}{.3pt}\rule{.3pt}{1ex}
\rule{.3em}{.3pt}\rule{.3pt}{1ex}}\mkern2mu}}

\title{Syntactic complexity of regular ideals}

\author{Janusz~Brzozowski\inst{1} \and Marek Szyku{\l}a\inst{2} and Yuli Ye\inst{3}}

\titlerunning{Syntactic Complexity of Regular Ideals}

\authorrunning{J. Brzozowski and M. Szyku{\l}a and Y. Ye}   

\institute{David R. Cheriton School of Computer Science, University of Waterloo, \\
Waterloo, ON, Canada N2L 3G1\\
\{{\tt brzozo@uwaterloo.ca}\}
\and
Institute of Computer Science, University of Wroc{\l}aw,\\
Joliot-Curie 15, PL-50-383 Wroc{\l}aw, Poland\\
\{{\tt msz@cs.uni.wroc.pl}\}
\and
Department of Computer Science, University of Toronto,\\
Toronto, ON,  Canada M5S 3G4\\
\{ { \tt yuli.ye@gmail.com} \}\\
Yuli Ye's present address: \\ Wish.com, San Francisco CA, 94111, USA
}

\begin{document}
\maketitle
\begin{abstract}
The state complexity of a regular language is the number of states in a minimal deterministic finite automaton accepting the language.
The syntactic complexity of a regular language is the cardinality of its syntactic semigroup.
The syntactic complexity of a subclass of regular languages is the worst-case syntactic complexity taken as a function of the state complexity $n$ of languages in that class.
We prove that $n^{n-1}$, $n^{n-1}+n-1$, and $n^{n-2}+(n-2)2^{n-2}+1$ are tight upper bounds on the syntactic complexities of right ideals and prefix-closed languages, left ideals and suffix-closed languages, and two-sided ideals and factor-closed languages, respectively.
Moreover, we show that the transition semigroups meeting the upper bounds for all three types of ideals are unique, and the numbers of generators (4, 5, and 6, respectively) cannot be reduced.
\medskip

\noin
\end{abstract}

\section{Introduction}

Formal definitions of the concepts introduced in this section are given in Section~\ref{sec:prelim}.
For now we assume that the reader is familiar with basic properties of regular languages and finite automata as covered in~\cite{Per90,Yu97}, for example.

There are two fundamental congruence relations in the theory of regular languages: the Nerode (right) congruence~\cite{Ner58}, and the Myhill congruence~\cite{Myh57}. In both cases, a language is regular if and only if it is a union of congruence classes of a congruence of finite index.
The Nerode congruence leads to the definitions of left quotients of a language and the minimal deterministic finite automaton (DFA) recognizing the language, and the Myhill congruence, to the definitions of the syntactic semigroup of the language.

The \emph{state complexity} of a language is the number of states in a minimal DFA recognizing the language. This concept has been studied extensively; for surveys and references see~\cite{Brz10,Yu01}.
The \emph{syntactic complexity} of a regular language is the cardinality of its \emph{syntactic semigroup}, which is isomorphic to the \emph{transition semigroup} of a minimal DFA recognizing the language, where the transition semigroup is the semigroup of transformations of the set of states of the DFA 
induced by non-empty words.

Syntactic complexity does not refine state complexity, for there exist languages with the same syntactic complexity but different state complexities. However, it ofter helps to distinguish among languages with the same state complexity.
For example, the DFAs in Fig.~\ref{fig:3automata} all have the same alphabet, are all minimal, and all have state complexity three. However, the syntactic complexity of $\cD_1$ is 3, that of $\cD_2$ is 9, and that of $\cD_3$ is 27. 

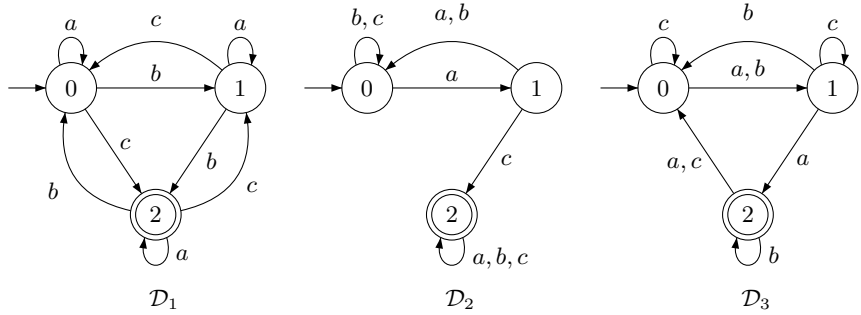
\begin{figure}[hbt]
\unitlength 8pt\footnotesize
\begin{center}\begin{picture}(38,14)(0,-6)
\gasset{Nh=2.4,Nw=2.4,Nmr=1.2,ELdist=0.35,loopdiam=1.2}

\node[Nframe=n](D1)(6.4,-6){$\cD_1 $}

\node(01)(2,4){$0$}\imark(01)
\node(11)(10,4){$1$}
\node(21)(6,-2){$2$}\rmark(21)

\drawedge(01,11){$b$}
\drawedge(01,21){$c$}
\drawedge(11,21){$b$}

\drawedge[curvedepth=-2.2,ELdist=-1.2](11,01){$c$}
\drawedge[curvedepth=-2.2,ELdist=-1.2](21,11){$c$}
\drawedge[curvedepth= 2.2,ELdist=.8](21,01){$b$}

\drawloop(01){$a$}
\drawloop(11){$a$}

\drawloop[loopangle=270,ELpos=25](21){$a$}

\node[Nframe=n](D2)(20.4,-6){$\cD_2$}
\node(02)(16,4){$0$}\imark(02)
\node(12)(24,4){$1$}
\node(22)(20,-2){$2$}\rmark(22)

\drawedge(02,12){$a$}
\drawedge(12,22){$c$}

\drawedge[curvedepth=-2.2,ELdist=-1.8](12,02){$a,b$}

\drawloop(02){$b,c$}
\drawloop[loopangle=270,ELpos=25](22){$a,b,c$}

\node[Nframe=n](D3)(34.4,-6){$\cD_3 $}
\node(03)(30,4){$0$}\imark(03)
\node(13)(38,4){$1$}
\node(23)(34,-2){$2$}\rmark(23)

\drawedge(03,13){$a,b$}
\drawedge(13,23){$a$}
\drawedge(23,03){$a,c$}

\drawedge[curvedepth=-2.2,ELdist=-1.8](13,03){$b$}

\drawloop(03){$c$}
\drawloop(13){$c$}
\drawloop[loopangle=270,ELpos=25](23){$b$}
\end{picture}\end{center}
\caption{DFAs with various syntactic complexities.}\label{fig:3automata}
\end{figure}

The problem we study in this paper is the following: Given a language belonging to a subclass of the class of regular languages -- for example, the subclass of finite languages or prefix-free languages (prefix-codes) -- what is the maximal size of the syntactic semigroup of that language? 
Equivalently, given a minimal DFA of a language in the subclass, what is the maximal size of the transition semigroup of the DFA?
A secondary problem is to find a minimal set of generators for the maximal semigroup.

Syntactic complexity has been studied in several subclasses of regular languages other than ideals: 
prefix-, suffix-, bifix-, and factor-free languages~\cite{BLY12,BrSz15};
star-free languages~\cite{BLL12,BrSz14b};
$R$- and $J$-trivial languages~\cite{BrLi15}; finite/cofinite and reverse definite languages~\cite{BLL12}.
This problem can be quite challenging, depending on the subclass; in the present case it is easy for right ideals but much more difficult for left- and two-sided ideals (defined below).

As syntactic complexity bounds the maximal size of the transition semigroup, it provides a natural bound on the time and space complexities of algorithms dealing with transition semigroups. For example, a simple algorithm determining whether the language of a given DFA is star-free \cite{McSe71} -- meaning that it can be generated from finite languages by using only Boolean operations and product (concatenation), but not star; equivalently, its syntactic semigroup is group-free, that is, has no non-trivial subgroups -- requires the enumeration of all transformations and checking whether they do not contain non-trivial cycles.

Maximal transition semigroups also play an important role in the study of \emph{most complex} languages~\cite{Brz13} belonging to a given subclass.
These are languages that meet all the upper bounds on the state complexities of Boolean operations, product, star, and reversal, have maximal syntactic semigroups and most complex atoms~\cite{BrTa14}.

In contrast to the \emph{syntactic monoid} of the language, the syntactic semigroup may or may not contain the neutral element (the identity transformation). The presence of letters acting as identity is often important in the case of state complexity of binary operations.
Moreover, the syntactic semigroup is more suitable to characterize some classes of languages, which have a description in terms of semigroups.
For example, in the class of (co)finite languages all transformations must admit a certain linear order of the states \cite{CCSY01}, and the identity transformation cannot be present; the latter condition would not be distinguished by the syntactic monoid.

In this paper we study the syntactic complexities of \emph{right ideals} (satisfying the equation $L=L\Sigma^*$), \emph{left ideals} (satisfying $L=\Sigma^* L$), and 
\emph{two-sided ideals} 
(satisfying $L=\Sigma^* L \Sigma^*$).
Ideals are fundamental objects in semigroup theory. 
They appear in the theoretical computer science literature 
in 1965~\cite{PaPe65} 
and continue to be of interest. 
Ideal languages are special cases of \emph{convex languages} (see e.g.~\cite{BSX10}), and they are complements 
of prefix-, suffix-, factor-, and subword-closed languages.
Besides being of theoretical interest, ideals also play a role in algorithms for pattern matching. 
For this application, a \emph{text} is represented by a word $w$ over some alphabet $\Sig$. 
A \emph{pattern} is language $L$ over $\Sig$.
An occurrence of a pattern represented by $L$ in text $w$ is a triple $(u,x,v)$ such that $w=uxv$ and $x$ is in~$L$.
Searching text $w$ for words in $L$ is equivalent to looking for prefixes of $w$ that belong to the language $\Sig^*L$, which is the left ideal generated by $L$, or looking for factors of $w$ that belong to $\Sig^*L \Sig^*$.

The state complexity of operations on the classes of ideal languages was studied by Brzozowski, Jir\'askov\'a and Li~\cite{BJL13}.
The same problem for the classes of 
prefix-, suffix-, factor-, and subword-closed languages was studied by Han and K.~Salomaa~\cite{HaSa09}, Han, K.~Salomaa, and Wood~\cite{HSW09}, and Brzozowski, Jir\'askov\'a and Zou~\cite{BJZ14}.
We refer the reader to these papers for a discussion of past work on this topic and additional references.

The set of all $n^n$ transformations of a set $Q_n$ of $n$ elements is a monoid under composition of transformations, with identity as the unit element. 
In 1970, Maslov~\cite{Mas70} dealt with the generators of the semigroup of all transformations in the setting of finite automata. 
Holzer and K\"onig~\cite{HoKo04}, and independently Krawetz, Lawrence, and Shallit~\cite{KLS05b} studied the syntactic complexity of unary and binary regular languages.
Recently, syntactic complexity has been studied in several subclasses of regular languages other than ideals: 
prefix-, suffix-, bifix-, and factor-free languages~\cite{BLY12,BrSz15};
star-free languages~\cite{BLL12,BrSz14b};
$R$- and $J$-trivial languages~\cite{BrLi15}.

We define our terminology and notation in Section~\ref{sec:prelim}, and give some basic properties of syntactic complexity in Section~\ref{sec:basic}. 
The syntactic complexities of right, left, and two-sided ideals are treated in Sections~\ref{sec:right}--\ref{sec:2sided}, and Section~\ref{sec:conclusions} concludes the paper.
As mentioned above, closed languages are complements of ideal languages.
Since syntactic complexity is preserved under complementation, our proofs are for ideals only.
The syntactic complexity of all-sided ideals remains open.

In the proof for the upper bounds for left and two-sided ideals we use the method of injective function, which is generally applicable for other subclasses of regular languages (see \cite{BrSz15} for suffix-free and~\cite{SzWi16} for bifix-free languages).
The proofs presented here are the first that apply this method to syntactic complexity.

A part of the results in this paper previously appeared in conference proceedings:
In 2011 in~\cite{BrYe11} syntactic complexity of right ideals was established and lower bounds for the classes of left and two-sided ideals were presented.
In 2014 in~\cite{BrSz14a} incomplete proofs of the upper bounds for syntactic complexity of left and two-sided ideals were presented.

\section{Preliminaries}\label{sec:prelim}

If $\Sig$ is a non-empty finite alphabet, then $\Sig^*$ is the free monoid generated by $\Sig$, and $\Sig^+$ is the free semigroup generated by $\Sig$.  A \emph{word} is any element of $\Sig^*$, and the empty word is $\eps$. The length of a word $w\in \Sig^*$ is $|w|$.
A \emph{language} over $\Sig$ is any subset of $\Sig^*$. 

If $w=uxv$ for some $u,v,x\in\Sigma^*$, then $u$ is a {\em prefix\/} of $w$, $v$ is a {\em suffix\/} of $w$, and $x$ is a {\em factor\/} of $w$.
A prefix or suffix of $w$ is also a factor of $w$.
If $w=u_1v_1u_2v_2\cdots u_kv_ku_{k+1}$, where the $u_i$ and $v_i$ are in $\Sig^*$, then
$v_1v_2\cdots v_k$ is a \emph{subword} of $w$.
A~language $L$ is {\it prefix-closed\/} if $w\in L$ implies that
every prefix of $w$ is also in~$L$. 
In an analogous way, we define \emph{suffix-closed}, \emph{factor-closed}, and \emph{subword-closed}. 
We refer to all four types as \emph{closed languages}.

The \emph{shuffle $u\shu v$ of two words} $u,v\in\Sig^*$ is defined as follows:
$$ u\shu v= \{u_1v_1\cdots u_kv_k \mid u=u_1\cdots u_k,  v=v_1\cdots v_k,
u_1,\ldots,u_k,v_1,\ldots, v_k\in \Sig^*\}.$$
The \emph{shuffle of two languages} $K$ and $L$ is defined by
$$K\shu L=\bigcup_{u\in K, v\in L} u\shu v.$$

A~language $L\subseteq\Sig^*$ is a \emph{right ideal} (respectively, \emph{left ideal}, \emph{two-sided ideal}, \emph{all-sided ideal}) if it is non-empty and satisfies $L=L\Sig^*$ (respectively, $L=\Sig^*L$, $L=\Sig^*L\Sig^*$,
$L=\Sig^*\shu L$).
We refer to all four of these types as \emph{ideal languages} or simply \emph{ideals}.

A \emph{transformation} of a set $Q_n$ of $n$ elements is a mapping of $Q_n$ \emph{into} itself, whereas a \emph{permutation}
of $Q_n$ is a mapping of $Q_n$ \emph{onto} itself.
In this paper we consider only transformations of finite sets, and we assume
without loss of generality that $Q_n=\{0,1,\ldots, n-1\}$.
An arbitrary transformation has the form
$$
t=\left( \begin{array}{ccccc}
0 & 1 &   \cdots &  n-2 & n-1 \\
q_0 & q_1 &   \cdots &  q_{n-2} & q_{n-1}
\end{array} \right ),
$$
where $q_k \in Q_n$ for $0\le k\le n-1$.
The image of element $q$ under transformation $t$ is denoted by $qt$.
The \emph{identity} transformation $\tid$ maps each element to itself.
For $k\ge 2$, a transformation (permutation) $s$ of a set $P=\{p_0,p_1,\ldots,p_{k-1}\} \subseteq Q_n $ is a \emph{$k$-cycle}
if $p_0s=p_1, p_1s=p_2,\ldots,p_{k-2}s=p_{k-1},p_{k-1}s=p_0$.
If a transformation $t$ on $Q_n$ acts on $P \subseteq Q_n$ like a $k$-cycle then $t$ is said to \emph{have a $k$-cycle}.
A $k$-cycle is denoted by $(p_0,p_1,\ldots,p_{k-1})$ when it is viewed as a transformation of $P$.
If $t$ is a transformation of $Q_n$, has a $k$-cycle $(p_0,p_1,\ldots,p_{k-1})$ of $P$, and acts as identity on $Q_n\setminus P$, then we denote $t$ also by $(p_0,p_1,\ldots,p_{k-1})$.
A~2-cycle $(p_0,p_1)$ is called a \emph{transposition}.
A transformation is \emph{constant} if it maps all states to a single state $q$; it is denoted by $(Q\to q)$.
A transformation that maps a single state $p$ to $q$ and keeps $Q \setminus \{p\}$ unchanged is denoted by $(p \to q)$.
If $w$ is a word of $\Sig^*$, the fact that $w$ induces transformation $t$ is denoted by 
$w\colon t$.
A~transformation mapping $p$ to $q_p$ for $p=0, \dots, n-1$ is sometimes denoted by
$[q_0, \dots,q_{n-1}]$.

The following facts are well-known \cite{Pic38,Sie35}:

\begin{proposition}\label{prop:piccard}
The complete transformation monoid $T_n$ of size $n^n$ can be generated by any cyclic
permutation of $n$ elements together with a transposition and a singular (non-invertible) transformation $r=(n-1 \to 0)$ of rank (image size) $n-1$. In particular, $T_n$ can be generated by $(0,1,\ldots,n-1)$, $(0,1)$ and $(n-1 \to 0)$. Moreover, $T_n$ cannot be generated by fewer than three generators for $n \ge 3$.
\end{proposition}

The \emph{left quotient}, or simply \emph{quotient,} of a language $L$ by a word $w$ is the language $w^{-1}L = \{x\in \Sig^*\mid wx\in L \}$. 
An equivalence relation $\sim$ on $\Sig^*$ is a \emph{left congruence} if, for all $x,y\in \Sig^*$, 
$x\sim y {\Lra} ux \sim uy, \mbox{ for all } u\in\Sig^*.$
It is a \emph{right congruence} if, for all $x,y\in \Sig^*$,
$x\sim y {\Lra} xv \sim yv, \mbox{ for all } v\in\Sig^*.$
It is a \emph{congruence} if it is both a left and a right congruence. 
Equivalently, $\sim$ is a congruence if 
$x\sim y {\Lra} uxv \sim uyv, \mbox{ for all } u,v\in\Sig^*.$

For any language $L\subseteq \Sig^*$, define the \emph{Nerode (right) congruence}~\cite{Ner58} $\raL$ of $L$ by
\begin{equation}
x \,{\raL} \, y \mbox{ if and only if } xv\in L  \Leftrightarrow yv\in L, \mbox { for all } v\in\Sig^*.
\end{equation}
Evidently, $x^{-1}L = y^{-1} L $ if and only if $x \, \raL \,  y$.
Thus, each equivalence class of this congruence corresponds to a distinct quotient of $L$.
Let $K=\{K_0,\dots,K_{n-1}\}$ be the set of quotients of a regular language $L$; by convention, we let $K_0=L=\eps^{-1} L$.
The number of distinct quotients of $L$ is the \emph{quotient complexity} $\kappa(L)$ of $L$.

The \emph{Myhill congruence}~\cite{Myh57} $\lraL$ of $L$ is defined by
\begin{equation}
x \, \lraL\,  y \mbox{ if and only if } uxv\in L  \Leftrightarrow uyv\in L\mbox { for all } u,v\in\Sig^*.
\end{equation}
This congruence is also known as the \emph{syntactic congruence} of $L$.
The semigroup $\Sig^+/ \lraL$ of equivalence classes of the relation $\lraL$, is the \emph{syntactic semigroup} of $L$, and 
$\Sig^*/ \lraL$ is the \emph{syntactic monoid} of~$L$. 
The \emph{syntactic complexity} $\sig(L)$ of $L$ is the cardinality of its syntactic semigroup.

A~\emph{deterministic finite automaton (DFA)} is a quintuple $\cD=(Q, \Sig, \delta, q_0,F)$, where $Q$ is a finite, non-empty set of \emph{states}, $\Sig$ is a finite non-empty \emph{alphabet}, $\delta\colon Q\times \Sig\to Q$ is the \emph{transition function}, $q_0\in Q$ is the \emph{initial state}, and $F\subseteq Q$ is the set of \emph{final states}.
By the \emph{language of a state} $q$ of $\cD$
we mean the language $K_q$ accepted by the automaton $(Q,\Sigma,\delta,q,F)$.
States $p$ and $q$ are \emph{equivalent} if $K_p = K_q$.
A state $q$ is \emph{reachable} if $\delta(q_0,w)=q$ for some $w\in \Sig^*$.
A~DFA is \emph{minimal} if every state is reachable and no two states are equivalent.

The \emph{quotient automaton} of $L$ is 
$\cD=(K, \Sig, \delta, L,F)$, where $\delta(K_q,a)=a^{-1}K_q$, 
and $F=\{K_q \mid \eps\in K_q\}$.
The quotient automaton is always minimal, and so quotient complexity is the same as state complexity.

The \emph{transition semigroup} of a DFA is the set of transformations induced by words of $\Sig^+$ on the set of states. The transition semigroup of the quotient DFA of $L$ is isomorphic to the syntactic semigroup of~$L$~\cite{Pin97}.

\section{Syntactic Complexity of Languages with Special Quotients}
\label{sec:basic}

We now present some basic properties of syntactic complexity.
\begin{proposition}\label{prop:basicbounds}
For any $L\subseteq \Sig^*$ with $\kappa(L)=n>1$, $n-1\le \sig(L)\le n^n$. 
\end{proposition}
\begin{proof}
Let $\cD=(K, \Sig, \delta, L,F)$ be the quotient automaton of $L$.
Since every state other than $L$ has to be reachable from the initial state $L$ by a non-empty word, there must be at least $n-1$ transformations.
If $\Sig=\{a\}$ and $L=a^{n-1}a^*$, then  $\kappa(L)=n$, and $\sigma(L)=n-1$;
so the lower bound $n-1$ is achievable. The upper bound is $n^n$, and by Proposition~\ref{prop:piccard} this upper bound is achievable if $|\Sig|\ge 3$.
\end{proof}

If one of the quotients of $L$ is $\emp$ (respectively, $\{\eps\}$, $\Sig^*$, $\Sig^+$), then we say that $L$ \emph{has $\emp$} (respectively, $\{\eps\}$, $\Sig^*$, $\Sig^+$).
A quotient $w^{-1}L$ of a language $L$ is \emph{\ur{}}~\cite{Brz10} if $x^{-1}L=w^{-1}L$ implies that $x=w$.
If $(wa)^{-1}L$ is \ur{} for $a\in \Sig$, then so is $w^{-1}L$. 
Thus, if $L$ has a \ur{} quotient, then $L$ itself is \ur{} by $\eps$, \ie, a minimal automaton of $L$ is \emph{non-returning}~\cite{HaSa09}.

\begin{theorem}[Special Quotients]
\label{thm:specialquotients}
Let $L\subseteq \Sig^*$ and let $\kappa(L)=n\ge1$. 
\be
\item
If $L$ has $\emp$ or $\Sig^*$, then $\sig(L)\le n^{n-1}$.
\item
If $L$ has $\{\eps\}$ or $\Sig^+$, then $\sig(L)\le n^{n-2}$.
\item
If $L$ is uniquely reachable, then $\sig(L)\le (n-1)^n$.
\item
If $w^{-1}L$ is uniquely reachable by $w \in \Sig^*$ with $0 \le |w|\le n-1$, then $\sig(L)\le |w|+(n-1-|w|)^{ n }$.
\ee
Moreover, all the bounds shown in Table~\ref{tab:special} hold.

\begin{table}[ht]
\caption{Upper bounds on syntactic complexity for languages with special quotients.
The abbreviation ``ur'' stands for ``uniquely reachable''. The $a$ in the last column is in $\Sig$.}\label{tab:special}
\begin{center}
$
\begin{array}{| c | c | c | c | c | c | c |}   
\hline
\hspace{.25cm}\emp \hspace{.25cm} &\hspace{.1cm}\Sig^*\hspace{.1cm} &
\hspace{.25cm}\{\eps\} \hspace{.25cm}&  \hspace{.1cm}\Sig^+\hspace{.1cm}  & \sigma(L)\le &  \mbox{if also $L$ is ur} & \mbox{if also $a^{-1}L$ is ur}
\rule[-0.2cm]{0cm}{0.7cm}
\\
\hline
\surd & &    &    & n^{n-1}  & (n-1)^{n-1} & 1+ (n-3)^{n-1}
\\
\hline
 & \surd &    &    & n^{n-1} & (n-1)^{n-1} & 1+ (n-3)^{n-1}
\\
\hline
\surd  &  & \surd  &    & n^{n-2} & (n-1)^{n-2} & 1+ (n-4)^{n-2}
\\
\hline
 & \surd  &    & \surd   &  n^{n-2} & (n-1)^{n-2} & 1+ (n-4)^{n-2}
\\
\hline
 \surd  & \surd   &    &  &  n^{n-2} & (n-1)^{n-2} & 1+ (n-4)^{n-2}
\\
\hline
 \surd  & \surd &      &\surd  &  n^{n-3} & (n-1)^{n-3} & 1+ (n-5)^{n-3}
\\
\hline
 \surd  &\surd   & \surd   &  & n^{n-3} & (n-1)^{n-3} & 1+ (n-5)^{n-3}
\\
\hline
 \surd  &\surd   & \surd   &\surd  & n^{n-4} & (n-1)^{n-4} & 1+ (n-6)^{n-4}
 \\
\hline
\end{array}
$
\end{center}
\end{table}
\end{theorem}
\begin{proof}
Suppose that $L\subseteq \Sig^*$, $n\ge 1$, and $\kappa(L)=n$.
\be
\item
Since $a^{-1}\emp=\emp$ for all $a\in\Sig$, there are only $n-1$ states in the quotient automaton with which one can distinguish two transformations. Hence there are at most $n^{n-1}$ transformations. If $L$ has $\Sig^*$, then $a^{-1}\Sig^*=\Sig^*$, for all $a\in\Sig$, and the same argument applies.
\item
Since $a^{-1}\{\eps\}=\emp$ for all $a\in\Sig$, $L$ has $\emp$ if $L$ has $\{\eps\}$. Now there are two states that do not contribute to distinguishing among different transformations. 
Dually, $a^{-1}\Sig^+=\Sig^*$ for all $a\in\Sig$, and the same argument applies.
\item
If $L$ is uniquely reachable then $w^{-1}L= L$ implies $w=\eps$. Thus $L$ does not appear in the image of any transformation by a word in $\Sig^+$, and there remain only $n-1$ choices for each of the $n$ states. 
\item
If $w^{-1}L$ is \ur{}, then so is $x^{-1}L$ for every prefix $x$ of $w$. Hence for each prefix $x$ of $w$, $x^{-1}L$ appears only in one transformation, and there are $|w|$ such transformations. All the other transformations map every quotient $x^{-1}L$ to $y^{-1}L$, where $y$ is not a prefix of $w$. 
Therefore there can be at most $(n-1-|w|)^{n}$ other transformations.
\ee
The remaining entries in Table~\ref{tab:special} are easily verified:
every transformation fixes $\emptyset$, $\Sigma^*$, maps $\{\varepsilon\}$ to $\emptyset$, and maps $\Sigma^+$ to $\Sigma^*$, so these quotients are removed from counting possible mappings for a quotient.
\end{proof}

\section{Right Ideals and Prefix-Closed Languages}
\label{sec:right}
In this section we prove that the syntactic complexity of right ideals is $n^{n-1}$. First we define a witness DFA that meets this bound.

\begin{definition}[Witness: Right Ideals]
\label{def:right}
For $n\ge 3$, let
$\cW_n=(Q_n, \Sig,\delta_\cW,0, \{n-1\}),$
be the DFA in which $\Sig=\{a,b,c,d\}$,
 $a\colon (0,\ldots,n-2)$, 
$b\colon (0,1)$, 
$c\colon (n-2 \to 0)$, and
$d\colon (n-2 \to n-1)$.
For $n=3$ inputs $a$ and $b$ induce the same transformation; hence $\Sig=\{a,c,d\}$ suffices. Furthermore, let $\cW_2=(Q_2,\{a,b\},\delta_\cW,0,\{1\})$, where  
 $a\colon (0\to 1)$, and $b\colon \one$, and let $\cW_1=(Q_1,\{a\},\delta_\cW,0,\{0\})$, where 
 $a\colon \one$.
 Let $L_n=L(\cW_n)$.
\end{definition}

The structure of the DFA of Definition~\ref{def:right} is shown in Fig.~\ref{fig:RWit} for $n\ge 3$.
\begin{figure}[ht]
\unitlength 12pt\small
\begin{center}\begin{picture}(29,6)(0,-1)
\gasset{Nh=2.1,Nw=2.1,Nmr=1.05,ELdist=0.3,loopdiam=1.2}
\node(0)(2,2){$0$}\imark(0)
\node(1)(7,2){$1$}
\node(2)(12,2){$2$}
\node[Nframe=n](qdots)(17,2){$\dots$}
\node(n-2)(22,2){$n-2$}
\node(n-1)(27,2){$n-1$}\rmark(n-1)

\drawedge(0,1){$a,b$}
\drawedge(1,2){$a$}
\drawedge(2,qdots){$a$}
\drawedge(qdots,n-2){$a$}
\drawedge(n-2,n-1){$d$}

\drawedge[curvedepth=2.1](n-2,0){$a,c$}
\drawedge[curvedepth=-2,ELdist=-1](1,0){$b$}

\drawloop(0){$c,d$}
\drawloop(1){$c,d$}
\drawloop(2){$b,c,d$}
\drawloop(n-2){$b$}
\drawloop(n-1){$a,b,c,d$}
\end{picture}\end{center}
\caption{Quotient DFA $\cW_n$ of a right ideal with $n^{n-1}$ transformations.}
\label{fig:RWit}
\end{figure}
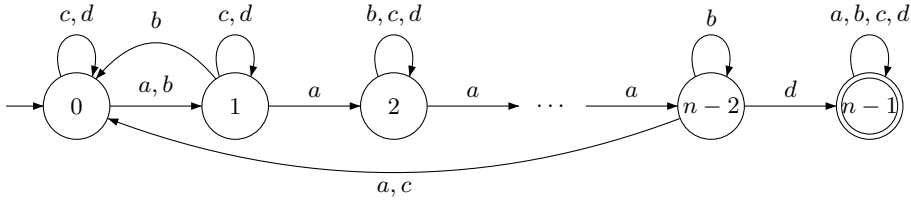

\begin{lemma}\label{lem:right}
The DFA $\cW_n$ of Definition~\ref{def:right} is minimal, accepts a right ideal, and has transition semigroup of size $n^{n-1}$.
\end{lemma}
\begin{proof}
If $n\le 2$ this is easily verified; here $L_1=\Sig^*$ and $L_2=\Sig^*a\Sig^*$. 

For $n\ge 3$, state $q$ with $0\le q \le n-2$ is non-final and accepts $a^{n-2-q}d$ and no other such state accepts this word.
Since $n-1$ is final, all states are distinguishable.
Since $\cW_n$ has exactly one final state and that state accepts $\Sig^*$, $L_n$ is a right ideal.

For the syntactic complexity, observe that inputs $a$, $b$, and $c$ restricted to $Q_{n-1}$ can induce any transformation of $Q_{n-1}$ (Proposition~\ref{prop:piccard}); hence all $(n-1)^{n-1}$ transformations that fix $n-1$ can be performed by $\cW_n$. 
Also observe that any transformation $(q \to n-1)$ for $q \in \{0,\ldots,n-3\}$ is induced by $a^{n-2-q} d a^{q+1}$;

Note that every transformation from the transition semigroup of $\cW_n$ fixes state $n-1$.
Let $t$ be any transformation such that $(n-1)t = n-1$.
There are $(n-1)^n$ such transformations, and we will show that all of them are generated.
Let $\{p_1,\ldots,p_k\}$ be the set of all states from $Q \setminus \{n-1\}$ that are mapped by $t$ to $n-1$.
Then $t$ can be generated by $(p_1 \to n-1)\cdots(p_k \to n-1)t'$, where $t'$ fixes $n-1$ and all states $p_i$, and acts as $t$ on the other states; thus it is a transformation of $Q_{n-1}$ if restricted to $Q_{n-1}$ and can be generated by $a$, $b$, and $c$.
\end{proof}

We are now in a position to state our main theorem of this section.

\begin{theorem}[Right Ideals and Prefix-Closed Languages]
\label{thm:right}
Suppose that $L\subseteq\Sig^*$ and $\kappa(L)=n$.
If $L$ is a right ideal or a prefix-closed language, then $\sig(L)\le n^{n-1}$.
This bound is tight for $n=1$ if $|\Sig|\ge 1$, for $n=2$ if $|\Sig|\ge 2$, for $n=3$ if $|\Sig|\ge 3$, and for $n\ge 4$ if $|\Sig|\ge 4$.
Moreover, the sizes of the alphabet cannot be reduced.
\end{theorem}
\begin{proof}
If $L$ is a right ideal, it has $\Sig^*$ as a quotient. By Theorem~\ref{thm:specialquotients}, $\sig(L_n)\le n^{n-1}$. 
By Lemma~\ref{lem:right} the languages of Definition~\ref{def:right} meet this bound.

It is easy to verify that the alphabet cannot be smaller if $n \le 3$. Let $n \ge 4$.
The transition semigroup restricted to $Q_{n-1}$ contains all transformations $Q_{n-1} \to Q_{n-1}$.
From Proposition~\ref{prop:piccard} there must be three generators of these transformations, say $a,b,c$.
They cannot map any state from $Q_{n-1}$ to $n-1$. Thus we need one more generator, say $d$, which maps a state from $Q_{n-1}$ to $n-1$.

Since prefix-closed languages are complements of right ideals and the syntactic complexity is preserved by complementation, the result is the same for prefix-closed languages.
\end{proof}

\begin{remark}
A maximal transition semigroup of the quotient DFA of a right ideal contains all transformations of $Q_n$ that fix state $n-1$. Hence there is only one maximal transition semigroup for right ideals.
\end{remark}

\section{Left Ideals and Suffix-Closed Languages}
\label{sec:left}
\subsection{Basic Properties}
Let $Q_n=\{0,\ldots,n-1\}$, 
let $\cD_n=(Q_n, \Sigma_\cD, \delta_\cD, 0,F)$ be a minimal DFA, and let $T_n$ be its transition semigroup.
Consider the sequence $(0,0t,0t^2,\dots)$ of states obtained by applying transformation $t\in T_n$ repeatedly, starting with the initial state.
Since $Q_n$ is finite, there must eventually be a repeated state, that is, there must exist $i$ and $j$ such that $0,0t,\dots,0t^i,0t^{i+1},\dots, 0t^{j-1}$ are distinct, but $0t^j=0t^i$;
the integer $j-i$ is the \emph{period} of $t$.
If the period  is $1$,  $t$ is said to be \emph{initially aperiodic};
then the sequence is $0,0t,\dots, 0t^{j-1}=0t^j$.

\begin{lemma}\label{lem:aperiodic}
If $\cD_n$ is the quotient DFA of a left ideal, all the transformations in $T_n$ are initially aperiodic, and no state of $\cD_n$ is empty.
\end{lemma}
\begin{proof}
Suppose that $w$ induces a transformation $t$ such that $p_i = 0t^i = 0t^j = p_j$ for some $i < j$, where $j-i \ge 2$; thus $t$ is not initially aperiodic.
Since $\cD_n$ is minimal, states $p_i$ and $p_{j-1}$ must be distinguishable, say by word $x\in\Sig^*$.
If $w^ix\in L$, then $w^{j-1} x=w^iw^{j-i-1}x=w^{j-i-1}(w^ix) \not\in L$, contradicting the assumption that $L$ is a left ideal.
If $w^{j-1}x\in L$, then $w^jx=w(w^{j-1}x) \not\in L$, again contradicting that $L$ is a left ideal.

For the second claim, we know that a left ideal is non-empty by definition. So suppose that $w\in L$. If $L$ has the empty quotient, say $x^{-1}L=\emp$, then $xw\not\in L$, which is a contradiction.
\end{proof}

\begin{example}
Note that the conditions of Lemma~\ref{lem:aperiodic} are not sufficient. For $\Sig=\{a,b\}$, the language $L=b\cup \Sig^*a$ satisfies the conditions, but is not a left ideal because $b\in L$ but $ab\not\in L$.
Its quotient automaton is shown in Fig.~\ref{fig:notideal}.

If the final state is 2 instead of 1, the language becomes $L'=\Sig\Sig^*b=\Sig^*\Sig b$, which \emph{is} a left ideal.
The languages $L$ and $L'$  have the same syntactic semigroup, but one is a left ideal while the other is not.
\end{example}

\begin{figure}[ht]
\unitlength 10pt\small
\begin{center}\begin{picture}(19,6)(0,8)
\gasset{Nh=2.5,Nw=2.5,Nmr=1.25,ELdist=0.5,loopdiam=1.5}

\node(0)(2,10){0}\imark(0)
\node(1)(8,10){$1$}\rmark(1)
\node(2)(14,10){$2$}
\drawedge(0,1){$a,b$}
\drawedge(1,2){$b$}
\drawloop(1){$a$}
\drawloop(2){$b$}
\drawedge[curvedepth=2](2,1){$a$}
\end{picture}\end{center}
\caption{Quotient DFA of a language that is not a left ideal.}
\label{fig:notideal}
\end{figure}
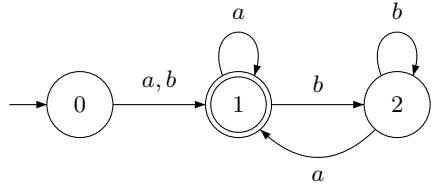

\begin{remark}[\cite{BJL13}]\label{rem:left-ideals_xy}
A language $L\subseteq \Sig^*$ is a left ideal if and only if for all $x,y\in\Sig^*$, $y^{-1} L \subseteq (xy)^{-1}L$.
Hence, if $x^{-1} L \neq L$, then $L\subset x^{-1} L$ for any $x\in\Sig^+$.
\end{remark}

It is useful to restate this observation it terms of the states of $\cD_n$.
For DFA $\cD_n$ and states $p,q \in Q_n$, we write $p \prec q$ if $K_p \subset K_q$. 

\begin{remark}\label{rem:left-ideals_xy2}
A DFA $\cD_n$ is a minimal DFA of a left ideal if and only if for all $s,t\in T_n \cup \{\tid\}$, $0t\preceq 0st$.
If $0t\neq 0$, then $0\prec 0t$ for any $t\in T_n$.
Also, if $r\in Q_n$ has a $t$-predecessor, that is, if there exists $q\in Q_n$ such that $qt=r$, then $0t\preceq r$. 
(This follows because $q=0s$ for some transformation $s$ since $q$ is reachable from 0; hence $0\preceq q$ and $0t \preceq qt =r$.)
In particular, if $r$ appears in a cycle of $t$ or is a fixed point of $t$, then $0t\preceq r$.
\end{remark}

We consider chains of the form $K_{i_1}\subset K_{i_2}\subset\dots \subset K_{i_h}$,
where the $K_{i_j}$ are quotients of $L$. 
If $L$ is a left ideal, the smallest element of any maximal-length chain is always $L$.
Alternatively, we consider chains of states starting from 0 and strictly ordered by $\prec$.
Since $\subseteq$ is a partial order on the quotients $K_{i_j}$, by definition of $\prec$ we have the following:

\begin{proposition}\label{prop:chain}
For $t\in T_n$ and $p, q\in Q_n$, $p \prec q$ implies $pt \preceq qt$.
If $p\prec pt$, then $p \prec pt \prec \dots \prec pt^k = pt^{k+1}$ for some $k \ge 1$. 
Similarly, $p \succ q$ implies $pt \succeq qt$, and $p \succ pt$ implies $p \succ pt \succ \dots \succ pt^k = pt^{k+1}$ for some $k \ge 1$.
\end{proposition}

\subsection{Lower Bound}
We now show that the syntactic complexity of the following DFA of a left ideal is $n^{n-1}+n-1$.

\begin{definition}[Witness: Left Ideals]
\label{def:left}
For $n\ge 3$, let
$\cW_n =(Q_n,\Sig_\cW,\delta_\cW,0,\{n-1\}),$
be the DFA in which
$\Sig_\cW=\{a,b,c,d,e\}$, 
$a\colon (1,\ldots,n-1)$,
$b\colon (1,2)$,
$c\colon (n-1\to 1)$,
$d\colon (n-1\to 0)$,
and $e\colon (Q_n \to 1)$.
For $n=3$, $a$ and $b$ coincide, and we can use $\Sig_\cW=\{a,c,d,e\}$.
Also, let $\cW_2=(Q_2,\{a,b,c\},\delta_\cW,0,\{1\})$, where  
 $a\colon (0\to 1)$,  $b\colon \one$, and $c\colon (Q_2\to 1)$, and let $\cW_1=(Q_1,\{a\},\delta,0,\{0\})$, where 
 $a\colon \one$.
 Let $L_n=L(\cW_n)$.
\end{definition}

The structure of the DFA of Definition~\ref{def:left} is shown in Fig.~\ref{fig:left} for $n\ge 3$.

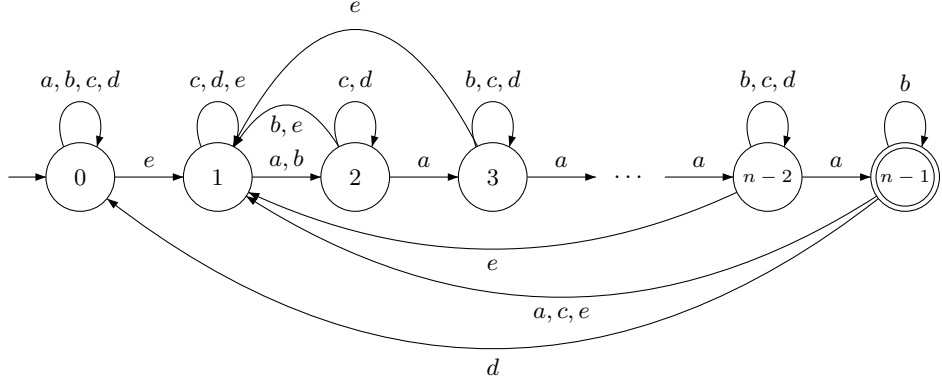
\begin{figure}[ht]
\unitlength 13pt\footnotesize
\begin{center}\begin{picture}(28,12)(0,0)
\gasset{Nh=2,Nw=2,Nmr=1,ELdist=0.3,loopdiam=1.2}
\node(0)(2,6){$0$}\imark(0)
\node(1)(6,6){$1$}
\node(2)(10,6){$2$}
\node(3)(14,6){$3$}
\node[Nframe=n](qdots)(18,6){$\dots$}
\node(n-2)(22,6){\scriptsize$n-2$}
\node(n-1)(26,6){\scriptsize$n-1$}\rmark(n-1)
\drawedge(0,1){$e$}
\drawedge[curvedepth=0,ELdist= 0.2](1,2){$a,b$}
\drawedge(2,3){$a$}
\drawedge(3,qdots){$a$}
\drawedge(qdots,n-2){$a$}
\drawedge(n-2,n-1){$a$}
\drawloop(0){$a,b,c,d$}
\drawloop(1){$c,d,e$}
\drawloop(2){$c,d$}
\drawloop(3){$b,c,d$}
\drawloop(n-2){$b,c,d$}
\drawloop(n-1){$b$}
\drawedge[curvedepth=-2.1](2,1){$b,e$}
\drawedge[curvedepth=-4.3,ELdist=-0.8](3,1){$e$}
\drawedge[curvedepth=2.1](n-2,1){$e$}
\drawedge[curvedepth=3.5](n-1,1){$a,c,e$}
\drawedge[curvedepth=5](n-1,0){$d$}
\end{picture}\end{center}
\caption{Quotient DFA $\cW_n$ of a left ideal with $n^{n-1}+n-1$ transformations.}
\label{fig:left}
\end{figure}

\begin{lemma}\label{lem:left}
The DFA of Definition~\ref{def:left} is minimal, accepts a left ideal, and has transition semigroup of size $n^{n-1}+n-1$ that contains all transformations fixing $0$ and all the constant transformations.
\end{lemma}
\begin{proof}
State $0$ does not accept $a^i$ for any $i$, whereas state $i$ with $1\le i\le n-2$ accepts $a^{n-1-i}$, and no other state of this type accepts this word. Since $n-1$ is the only final state, all states are distinguishable.

To prove that $L$ is a left ideal it suffices to show that for any $w\in L$, we also have $xw\in L$ for every $x\in\Sig$. This is obvious if $x\in\Sig\setminus \{e\}$. If $w\in L$, then $w$ has the form $w =uev$, where $\delta_\cW(0,u)=0$,
$\delta_\cW(0,ue)=1$, and $v$ is accepted from state 1. 
But  $\delta_\cW(0,eue)=1$,  and since $v$ is accepted from 1, we have $euev=ew\in L_n$.
Thus $L_n$ is a left ideal.

In $\cW_n$, the transformations induced by $a$, $b$, and $c$ restricted to $Q_n\setminus \{0\}$ generate all the transformations of the last $n-1$ states (Proposition~\ref{prop:piccard}).
Together with the transformation of $d$, they generate all transformations of $Q_n$ that fix 0, and the number of such transformations is $n^{n-1}$.
To see this, consider any transformation $t$ that fixes 0.
If some states from $\{1,\dots,n-1\}$ are mapped to 0 by $t$, we can map them first to $n-1$ and $n-1$ to one of them by the transformations of $a$, $b$, and $c$, and then map $n-1$ to 0 by the transformation of $d$.

Also the words of the form $e a^i$ for $i \in \{0,\ldots,n-2\}$ induce constant transformations $(Q_n \to i+1)$.
Hence the transition semigroup of $\cW_n$ contains all the constant transformations of $Q_n$ (where $(Q_n \to 0)$ has been already counted before).
Altogether, there are $n^{n-1}+n-1$ transformations in the transition semigroup of $\cW_n$.
\end{proof}

\begin{example}
One verifies that the maximal-length chains of quotients in $\cW_n$ have length~2.
On the other hand, for $n\ge 2$, let $\Sig=\{a,b\}$ and let $L=\Sig^*a^{n-1}$.
Then $L$ has $n$ quotients and the maximal-length chains are of length $n$.
\end{example}

We will see that the maximal length of chains of quotients is an important structural feature; in particular, to meet the bound for syntactic complexity by both left and two-sided ideals, the maximal length of the chains must be the smallest possible.

\subsection{Upper Bound}

The derivation of the upper bound $n^{n-1}+ n-1$ for left ideals is much more difficult that that for right ideals.
Our approach is as follows: We consider a minimal DFA $\cD_n=(Q_n, \Sigma_\cD, \delta_\cD, 0,F)$  of an arbitrary left ideal with $n$ quotients and let $T_n$ be the transition semigroup of $\cD_n$. 
We also deal with the witness DFA $\cW_n =(Q_n,\Sig_\cW,\delta_\cW,0,\{n-1\})$ of Definition~\ref{def:left} that has the same state set as $\cD_n$ and whose transition semigroup is $S_n$. We shall show that there is an injective mapping $f\colon T_n\to S_n$, and this will prove that $|T_n|\le |S_n|$.

\begin{remark}\label{rem:smalln}
If $n=1$, the only left ideal is $\Sig^*$ and the transition semigroup of its minimal DFA satisfies the bound
$1^0+1-1=1$.
If $n=2$, there are only three allowed transformations, since the transposition
$(0,1)$ is not initially aperiodic and is ruled out by Lemma~\ref{lem:aperiodic}.
Thus the bound $2^1+2-1=3$ holds.
\end{remark}

\begin{lemma}\label{lem:chain2}
If $n\ge 3$ and a maximal-length chain in $\cD_n$ strictly ordered by $\prec$ has length 2, then $|T_n|\le n^{n-1}+n-1$ and $T_n$ is a subsemigroup of $S_n$.
\end{lemma}
\begin{proof}
Consider an arbitrary transformation $t\in T_n$ and let $p=0t$. 
If $p=0$, then any state other than $0$ can possibly be mapped by $t$  to any one of the $n$ states; hence there are at most $n^{n-1}$ such transformations. All of these transformations are in $S_n$ by the proof of Lemma~\ref{lem:left}.

If $p\neq 0$, then $0\prec p$. Consider any state $q\not \in \{0,p\}$; by Remark~\ref{rem:left-ideals_xy2}, $p\preceq qt$.
If $p \neq qt$, then $p\prec qt$.
But then we have the chain $0\prec p\prec qt$ of length 3, contradicting our assumption.
Hence we must have $p=qt$, and so $t$ is the constant transformation $t= (Q_n\to p)$.
Since $p$ can be any one of the $n-1$ states other than 0, we have at most $n-1$ such transformations. 
Since all of these transformations are in $S_n$ by Lemma~\ref{lem:left}, $T_n$ is a subsemigroup of $S_n$.
\end{proof}

\begin{lemma}[Left Ideals, Suffix-Closed Languages]\label{lem:left_upper-bound}
If $n\ge 3$ and $L$ is a left ideal or a suffix-closed language with $n$ quotients, then its syntactic complexity is less than or equal to 
$n^{n-1}+n-1$.
\end{lemma}
\begin{proof}
It suffices to prove the result for left ideals, since suffix-closed languages are their complements.

For a transformation $t \in T_n$, consider the following cases:
\goodbreak
\smallskip

\noindent
\hglue 15pt  
{\bf Case 1:} $t \in S_n$. \\
Let $f(t) = t$; obviously $f(t)$ is injective.
\smallskip

\noindent
\hglue 15pt  
{\bf Case 2:} $t \not\in S_n$ and $0t^2 \neq 0t$.\\
Note that $t \not\in S_n$ implies $0 t \neq 0$ by Lemma~\ref{lem:left}.
Let $0t=p$.
We have $p=0t \prec 0tt=pt$ by Remark~\ref{rem:left-ideals_xy2}.
Let $p \prec \dots \prec p t^k = p t^{k+1}$ be the chain defined from $p$; this chain is of length at least 2.
Let $f(t)=s$, where $s$ is the transformation defined by

\begin{center} 
$0 s = 0, \quad p t^k s = p, \quad q s = q t \text{ for the other states } q\in Q_n.$
\end{center}
Transformation $s$ is shown in Fig.~\ref{fig:left-case2}, where the dashed transitions show how $s$ differs from $t$.
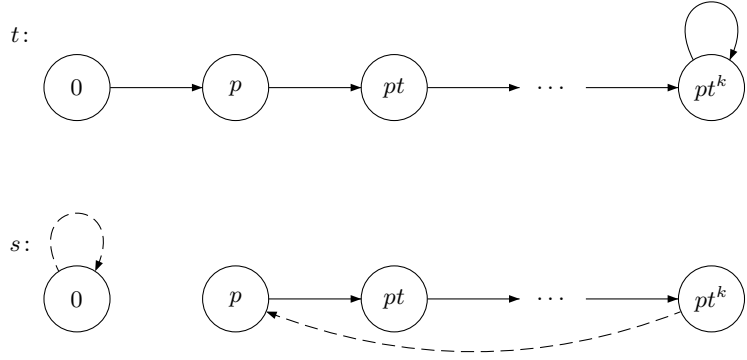
\begin{figure}[ht]
\unitlength 10pt\footnotesize
\begin{center}\begin{picture}(26,14)(0,0)
\gasset{Nh=2.5,Nw=2.5,Nmr=1.25,ELdist=0.5,loopdiam=2}
\node[Nframe=n](name)(0,12){$t\colon$}
\node(0)(2,10){0}
\node(p)(8,10){$p$}
\node(pt)(14,10){$pt$}
\node[Nframe=n](pdots)(20,10){$\dots$}
\node(pt^k)(26,10){$pt^k$}
\drawedge(0,p){}
\drawedge(p,pt){}
\drawedge(pt,pdots){}
\drawedge(pdots,pt^k){}
\drawloop[loopangle=90](pt^k){}

\node[Nframe=n](name)(0,4){$s\colon$}
\node(0')(2,2){0}
\node(p')(8,2){$p$}
\node(pt')(14,2){$pt$}
\node[Nframe=n](pdots')(20,2){$\dots$}
\node(pt^k')(26,2){$pt^k$}
\drawloop[loopangle=90,dash={.5 .25}{.25}](0'){}
\drawedge(p',pt'){}
\drawedge(pt',pdots'){}
\drawedge(pdots',pt^k'){}
\drawedge[curvedepth=2,dash={.5 .25}{.25}](pt^k',p'){}
\end{picture}\end{center}
\caption{Case~2 in the proof of Lemma~\ref{lem:left_upper-bound}.}
\label{fig:left-case2}
\end{figure}

By Lemma~\ref{lem:left}, $s \in S_n$.
However, $s\not\in T_n$, as it contains the cycle $(p, \ldots, p t^k)$ with states strictly ordered by $\prec$ in DFA $\cD_n$, which contradicts Proposition~\ref{prop:chain}.
Since $s \not\in T_n$, it is distinct from the transformations defined in Case~1.

In going from $t$ to $s$, we have added one transition ($0s=0$) that is a fixed point, and one ($pt^ks=p$) that is not.
Since only one non-fixed-point transition has been added, there can be only one cycle in $s$ with states strictly ordered by $\prec$.
Since $0$ cannot appear in this cycle, $p$ is its smallest element with respect to $\prec$.

Suppose now that $t'\neq t$ is another transformation that satisfies Case~2, that is, $0 t' = p' \neq 0$ and $p' t' \neq p'$; 
we shall show that $f(t) \neq f(t')$.
Define $s'$ for $t'$ as $s$ was defined for $t$.
For a contradiction, assume $s = f(t) = f(t')=s'$. 

Like $s$, $s'$ contains only one cycle strictly ordered by $\prec$, and $p'$ is its smallest element.
Since we have assumed that $s=s'$, we must have $p=0t=0t'=p'$ and the cycles in $s$ and $s'$ must be identical. 
In particular, $p t^k t = p t^k = p (t')^k t' = p (t')^k$.
For $q$ of $Q_n\setminus \{0,pt^k\}$, we have $qt=qs=qs'=qt'$.  
Hence $t = t'$---a contradiction. 
Therefore $t\neq t'$ implies $f(t)\neq f(t')$.
\smallskip

\noindent
\hglue 15 pt
{\bf Case 3:}  $t \not\in S_n$ and $0t^2 = 0t$. \\
As before, let $0t=p$. Consider any state $q\not\in \{0,p\}$; then
$0\prec q$ by Remark~\ref{rem:left-ideals_xy2} and $0t\preceq qt$ by Proposition~\ref{prop:chain}.
Thus either $p\prec qt$, or $p= qt$.
We consider the following sub-cases:
\smallskip

\noindent
\hglue 15 pt
$\bullet$ {\bf (a):} $t$ has a cycle.\\
Since $t$ has a cycle, take a state $r$ from the cycle; then $r$ and $rt$ are not comparable under $\preceq$ by Proposition~\ref{prop:chain}, and $p \prec r$ by Remark~\ref{rem:left-ideals_xy2}.
Let $f(t) = s$, where $s$ is the transformation shown in Figure~\ref{fig:left-case3a} and defined by
\begin{center}
$0 s = 0, \quad p s = r, \quad q s = q t \text{ for the other states } q \in Q_n.$
\end{center}

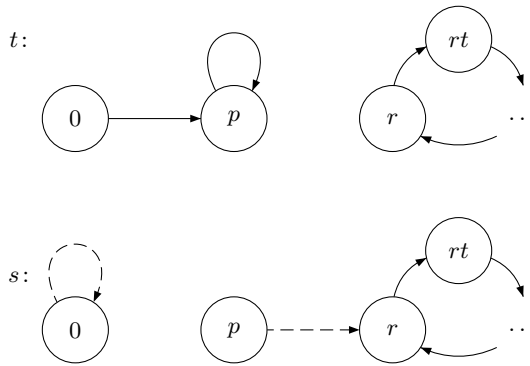
\begin{figure}[ht]
\unitlength 10pt\footnotesize
\begin{center}\begin{picture}(19,14)(0,0)
\gasset{Nh=2.5,Nw=2.5,Nmr=1.25,ELdist=0.5,loopdiam=2}
\node[Nframe=n](name)(0,13){$t\colon$}
\node(0)(2,10){0}
\node(p)(8,10){$p$}
\node(r)(14,10){$r$}
\node(rt)(16.5,13){$rt$}
\node[Nframe=n](rdots)(19,10){$\dots$}
\drawedge(0,p){}
\drawloop[loopangle=90](p){}
\drawedge[curvedepth=1](r,rt){}
\drawedge[curvedepth=1](rt,rdots){}
\drawedge[curvedepth=1](rdots,r){}
\node[Nframe=n](name)(0,4){$s\colon$}
\node(0')(2,2){0}
\node(p')(8,2){$p$}
\node(r')(14,2){$r$}
\node(rt')(16.5,5){$rt$}
\node[Nframe=n](rdots')(19,2){$\dots$}
\drawloop[loopangle=90,dash={.5 .25}{.25}](0'){}
\drawedge[dash={.5 .25}{.25}](p',r'){}
\drawedge[curvedepth=1](r',rt'){}
\drawedge[curvedepth=1](rt',rdots'){}
\drawedge[curvedepth=1](rdots',r'){}
\end{picture}\end{center}
\caption{Case~3(a) in the proof of Lemma~\ref{lem:left_upper-bound}.}
\label{fig:left-case3a}
\end{figure}

By Lemma~\ref{lem:left}, $s \in S_n$. 
Suppose that $s \in T_n$; since $p \prec r$, we have 
$r=ps \preceq rs=rt$ by the definition of $s$ and Proposition~\ref{prop:chain};
this contradicts that $r$ and $rt$ are not comparable.
Hence $s \not\in T_n$, and so $s$ is distinct from the transformations of Case~1.

We claim that $p$ is not in a cycle of $s$; this cycle would have to be
\begin{center}
 $p\stackrel{s }{\rightarrow} r \stackrel{s }{\rightarrow} rt \stackrel{s }{\rightarrow} \dots \stackrel{s }{\rightarrow} rt^{k-1}
\stackrel{s}{\rightarrow} p,
\text { that is, }
p\stackrel{s }{\rightarrow} r \stackrel{t }{\rightarrow} rt \stackrel{t }{\rightarrow} \dots \stackrel{t }{\rightarrow} rt^{k-1}
\stackrel{t}{\rightarrow} p,
$
\end{center}
\goodbreak
\noin
for some $k\ge 2$
because $r\neq p=pt$ and $rt\neq p$.
Since $p\prec r$  we have $p \prec rt$; but then we have a chain
$p\prec rt \prec \dots \prec rt^{k}=p$, contradicting Proposition~\ref{prop:chain}.

Since $p$ is not in a cycle of $s$, it follows that $s$ does not contain a cycle with states strictly ordered by $\prec$, as such a cycle would also be in $t$. So $s$ is distinct from the transformations of Case~2.

We claim there is a unique state $q$ such that (a) $0 \prec q \prec q s$, (b) $q s \not\preceq q s^2$.
First we show that $p$ satisfies these conditions: 
(a) holds because $ps=r$ and $p\prec r$;
(b) holds because $ps=r$, $ps^2=rt$ and $r$ and $rt$ are not comparable.
Now suppose that $q$ satisfies the two conditions, but $q \neq p$.
Note that $qs\neq p$, because $qs=p$ implies $qs=p\prec r=qs^2$, contradicting (b).
Since $q,q s \not\in \{0,p\}$, we have $q t = q s \not\preceq q s^2 = q t^2$. But Proposition~\ref{prop:chain} for $q \prec q t$ implies that $q t \preceq q t^2$---a contradiction.
Thus $p$ is the only state satisfying these conditions.

If $t' \neq t$ is another transformation satisfying the conditions of this case, we define $s'$ like $s$. Suppose that $s = f(t) = f(t') = s'$. Since both $s$ and $s'$ contain a unique state $p$ satisfying the two conditions above, we have $0 t = 0 t' = p$ and $p t = p t' = p$. 
Since the other states  are mapped by $s$ exactly as by $t$ and $t'$, we have $t = t'$.

\noindent
\hglue 15pt
$\bullet$ {\bf (b):} $t$ has no cycles and has a fixed point $r\neq p $.\\
Because $0\prec r$ by Remark~\ref{rem:left-ideals_xy2},
$0t\preceq rt$ by Proposition~\ref{prop:chain}.
If $r$ is a fixed point of $t$, then $p=0t \preceq rt=r$.
Since $r\neq p$, we have $p\prec r$.
Let $f(t)=s$, where $s$ is the transformation shown in Figure~\ref{fig:left-case3b} and defined by
\begin{center}
$0s=0, \quad q s = 0 \text{ for each fixed point } q\neq p,$
$q s = q t \text{ for the other states } q\in Q_n.$
\end{center}

By Lemma~\ref{lem:left}, $s \in S_n$.
Suppose that $s \in T_n$; because $p \prec r$, $ps=p$, $rs=0$, and $ps \preceq rs$ by Proposition~\ref{prop:chain}, we have $p \prec 0$, which is a contradiction.
Hence $s$ is not in $T_n$
and so is distinct from the transformations of Case~1. Also, $s$ maps at least one state other than $0$ to $0$, and so is distinct from the transformations of Case~2  and also from the transformations of Case~3(a).

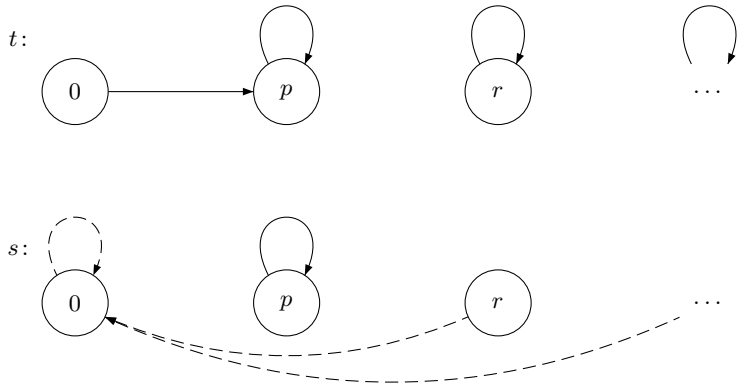
\begin{figure}[ht]
\unitlength 10pt\footnotesize
\begin{center}\begin{picture}(28,14)(0,0)
\gasset{Nh=2.5,Nw=2.5,Nmr=1.25,ELdist=0.5,loopdiam=2}
\node[Nframe=n](name)(0,12){$t\colon$}
\node(0)(2,10){0}
\node(p)(10,10){$p$}
\node(r)(18,10){$r$}
\node[Nframe=n](rdots)(26,10){$\dots$}
\drawedge(0,p){}
\drawloop[loopangle=90](p){}
\drawloop[loopangle=90](r){}
\drawloop[loopangle=90](rdots){}

\node[Nframe=n](name)(0,4){$s\colon$}
\node(0')(2,2){0}
\node(p')(10,2){$p$}
\node(r')(18,2){$r$}
\node[Nframe=n](rdots')(26,2){$\dots$}
\drawloop[loopangle=90,dash={.5 .25}{.25}](0'){}
\drawloop[loopangle=90](p'){}
\drawedge[curvedepth=2,ELpos=30,dash={.5 .25}{.25}](r',0'){}
\drawedge[curvedepth=3,ELpos=30,dash={.5 .25}{.25}](rdots',0'){}
\end{picture}\end{center}
\caption{Case~3(b) in the proof of Lemma~\ref{lem:left_upper-bound}.}
\label{fig:left-case3b}
\end{figure}

If $t'\neq t$ is another transformation satisfying the conditions of this case, we define $s'$ like $s$.
Now suppose that $s = f(t) = f(t')=s'$.
There is only one fixed point of $s$ other than $0$ ($p s = p$), and only one fixed point of $s'$ other than 0 ($p' s' = p'$); hence $0 t= p= p'= 0 t'$. 
By the definition of $s$, for each state $q\neq 0$ such that $q s = 0$, we have $q t = q$.
Similarly, for each state $q\neq 0$ such that $qs'=0$, we have $q t' = q$. 
Hence $t$ and $t'$ agree on these states.
Since the remaining states are mapped by $s$ exactly as they are mapped by $t$ and $t'$, we have $t = t'$.
Thus we have proved that $t\neq t'$ implies $f(t)\neq f(t')$. 
\smallskip

\noindent
\hglue 15 pt
$\bullet$ {\bf (c):} $t$ has no cycles, has no fixed point $r\neq p$, and there is a state $r$ 
such that $p\prec r$ with $rt=p$.\\
Let $f(t) = s$, where $s$ is the transformation shown in Figure~\ref{fig:left-case3c} and defined by
\begin{center}
$0 s = 0, \hspace{.2cm} p s = r, \hspace{.2cm} q s = 0 \text{ for each $q \succ p$ such that } q t = p,$

$\hspace{.2cm} q s = q t \text{ for the other states } q \in Q_n.$
\end{center}
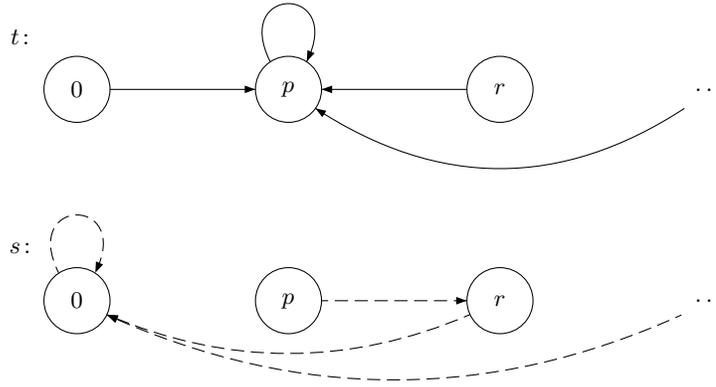
\begin{figure}[ht]
\unitlength 10pt\footnotesize
\begin{center}\begin{picture}(26,14)(0,-0.5)
\gasset{Nh=2.5,Nw=2.5,Nmr=1.25,ELdist=0.5,loopdiam=2}
\node[Nframe=n](name)(0,12){$t\colon$}
\node(0)(2,10){0}
\node(p)(10,10){$p$}
\node(r)(18,10){$r$}
\node[Nframe=n](rdots)(26,10){$\dots$}
\drawedge(0,p){}
\drawloop[loopangle=90](p){}
\drawedge(r,p){}
\drawedge[curvedepth=3](rdots,p){}

\node[Nframe=n](name)(0,4){$s\colon$}
\node(0')(2,2){0}
\node(p')(10,2){$p$}
\node(r')(18,2){$r$}
\node[Nframe=n](rdots')(26,2){$\dots$}
\drawloop[loopangle=90,dash={.5 .25}{.25}](0'){}
\drawedge[dash={.5 .25}{.25}](p',r'){}
\drawedge[curvedepth=2,ELpos=30,dash={.5 .25}{.25}](r',0'){}
\drawedge[curvedepth=3,ELpos=30,dash={.5 .25}{.25}](rdots',0'){}
\end{picture}\end{center}
\caption{Case~3(c) in the proof of Lemma~\ref{lem:left_upper-bound}.}
\label{fig:left-case3c}
\end{figure}

By Lemma~\ref{lem:left}, $s \in S_n$. 
Suppose that $s \in T_n$; because $p \prec r$, $p s = r$, $r s = 0$, and $r= p s \prec r s=0$ by Proposition~\ref{prop:chain}, we have $r \prec 0$---a contradiction. Hence $s \not\in T_n$ and $s$ is distinct from the transformations of Case~1.

Because $s$ maps at least one state other than $0$ to $0$ ($r s = 0$), it is distinct from the transformations of Case~2 and~3(a).
Also $s$ does not have a fixed point other than $0$, while the transformations of Case~3(b) have such a fixed point.

We claim that there is a unique state $q$ such that (a) $0 \prec q \prec q s$ and (b) $q s^2 = 0$. First we show that $p$ satisfies these conditions.
By assumption $0\prec p\prec r$ and  $rt=p$; also $rs=0$
by the definition of $s$.
Condition (a) holds because $0 \prec p \prec r = p s$, and (b) holds because $0 = r s = p s^2  $.

Now suppose that $0 \prec q \prec q s$, $q s^2 = 0$ and $q \neq p$.
Since $q s \neq 0$, we have $q s = q t$ by the definition of $s$. 
Because $q t$ has a $t$-predecessor, $p \preceq q t$ by Remark~\ref{rem:left-ideals_xy2}.
Also $q t = q s \neq p$, for $qs = p$ implies $0=qs^2=ps=r$---a contradiction.
Hence $p\prec q t$.
From $qt = qs$ and $q \prec q s$, we have $q \prec q t$.
Since $q s^2 = 0$ we have $(q t) s =0$ and so $(q t) t = p$, by the definition of $s$.
By Proposition~\ref{prop:chain}, from $q \prec q t $ we have $q t \preceq (q t) t = p$, contradicting $p \prec q t$.
So $q=p$.

If $t' \neq t$ is another transformation satisfying the conditions of this case, we define $s'$ like $s$. Suppose that $s = f(t) = t(t') = s'$. Since $s$ and $s'$ contain a unique state $p$ satisfying the two conditions above, we have $0 t = 0 t' = p$ and $p t = p t' = p$. Then $r$ and the states $q \succ p$ with $qt=p$ are determined by $p$, since they are precisely the states $q \succ p$ with $qs = 0$. Since the other states are mapped by $s$ exactly as by $t$ and $t'$, we have $t = t'$, and $f$ is again injective.
\smallskip

\hglue 5pt $\bullet$ {\bf All cases are covered:} \\ 
See the appendix for a list of the cases.
We need to ensure that any transformation $t$ fits in at least one case.
It is clear that $t$ fits in Case~1 or~2 or~3.
For Case~3, it is sufficient to show that if (i) $t \not\in S_n$ does not contain a fixed point $r\neq p$, and (ii) there is no state $r$ with $p \prec r$ and $r t = p$, then $t$ contains a cycle and so fits in Subcase~3(c).

First, if there is no $r$ such that $p\prec r$, we claim that $t$ is the constant transformation $(Q_n\to p)$, thus it fits in Case~1.
Consider any state $q \in Q_n$ such that $q t \neq p$.
Then $p\prec qt$ by Remark~\ref{rem:left-ideals_xy2}, contradicting that there is no state $r= qt$ such that $p\prec r$.

So let $t$ be a transformation that fits in Case~3 and satisfies~(i) and~(ii), and let $r$ be some state such that $p\prec r$.
Consider the sequence $r, rt, rt^2, \ldots$.
By Remark~\ref{rem:left-ideals_xy2}, $p \preceq rt^i$ for all $i \ge 0$.
If $rt^k = p$ for some $k \ge 1$, let $k$ be the smallest such number, then $rt^{k-1} \neq p$; we have $p \prec rt^{k-1}$ and $(rt^{k-1}) t = p$, contradicting (ii).
Since $p$ is the only fixed point by~(i), we have $rt^i \neq rt^{i-1}$ for all $i \ge 1$.
Since there are finitely many states, $rt^i = rt^j$ for some $i$ and $j$ such that $0 \le i < j-1$, and so the states $rt^i,rt^{i+1}, \ldots, rt^j=rt^i$ form a cycle.
\smallskip

We have shown that for every transformation $t$ in $T_n$ there is a corresponding transformation $f(t)$ in $S_n$, and $f$ is injective.
So $|T_n|\leq |S_n|=n^{n-1}+n-1$.
\end{proof}

Next we prove that $S_n$ is the only transition semigroup meeting the bound. It follows that minimal DFAs of left ideals with the maximal syntactic complexity have maximal-length chains of length $2$.

\begin{theorem}\label{thm:left-ideals_unique-maximal}
If $T_n$ has size $n^{n-1}+n-1$, then $T_n=S_n$. 
\end{theorem}
\begin{proof}
Consider a maximal-length chain of states strictly ordered by $\prec$ in $\cD_n$. If its length is 2, then by Lemma~\ref{lem:chain2}, $T_n$ is a subsemigroup of $S_n$. Thus only $T_n = S_n$ reaches the bound in this case.

Assume now that the length of a maximal-length chain is at least 3. 
Then there are states $p$ and $r$ such that $0 \prec p \prec r$.
Let $R=\{q \mid p\prec q\}$, and let $X=Q_n\setminus (R\cup \{0,p\})$.
We shall show that there exists a transformation $s$ that is in $S_n$ but not in $f(T_n)$.
To define $s$ we use  the constant transformation $u=(Q_n\to p)$ as an auxiliary transformation.
Let $0s=0$, $ps=r$, $rs=0$ for all $r\in R$, and $qs = qu=p$ for $q\in X$; these are precisely the rules  we used in Case 3(c) in the proof of Lemma~\ref{lem:left_upper-bound}.
By Lemma~\ref{lem:left}, $s \in S_n$.

It remains to be shown that there is no transformation $t\in T_n$ such that $s=f(t)$.
The proof that $s$ is different from the transformations $f(t)$ of Cases~1, 2, 3(a) and 3(b) is exactly the same as the corresponding proof in Case~3(c) following the definition of $s$.

It remains to verify that there is no $u' \in T_n$\ in Case~3(c) such that $f(u')=s$.
Suppose there is such a $u'$.  
Recall that states $p$ and $r$ satisfying $0 \prec p \prec r$ have been fixed by assumption. 
By the definition of $s$, state  $p$ satisfies the conditions (a) $0 \prec p \prec ps$ and (b) $ps^2 = 0$.
We claim that $p$ is the only state satisfying these conditions.
Indeed, if $q \neq p$ then either $qs = 0$, $q \not\prec qs=0$ and (a) is violated, or
$qs = p$,  $qs^2 = ps = r \neq 0$ and (b) is violated.
This observation is used in the proof of Case~3(c) to prove the claim below.

Both $u$ and $u'$ satisfy the conditions of Case~3(c), except that $u$ fails the condition $u\not\in S_n$. 
However, that latter condition is not used in the proof that 
if $u \neq u'$ and $u'$ satisfy the other conditions of Case~3(c), then $s' \neq s$, where $s'$ is the transformation obtained from $u'$ by the rules of $s$. Thus $s$ is also different from the transformations in $f(T_n)$ from Case~3(c).

Because $s \not\in f(T_n)$, $s \in S_n$ and $f(T_n) \subseteq S_n$,
the bound $n^{n-1}+n-1$ cannot be reached if the length of the maximal-length chains is not 2.
\end{proof}

\begin{proposition}\label{prop:left-ideals_maximal_generators}
For $n \ge 4$, the minimal number of generators of the transition semigroup $T_n$ is 5.
\end{proposition}
\begin{proof}
Since all transformations mapping 0 to a state in $Q_n \setminus \{0\}$ are constant transformations,
they must be generated by constant generators. Let $e$ be one of these generators.
Transition semigroup $T_n$ contains all transformations from $Q_n \setminus \{0\}$ to $Q_n \setminus \{0\}$ that fix 0.
By Proposition~\ref{prop:piccard} we need three generators for them, and the generators must fix 0; otherwise a generator would have to be constant, and the only constant transformation fixing 0 is $(Q_n \to 0)$. Let $a,b,c$ be three of these generators.
Finally, $T_n$ contains transformations mapping some states from $Q_n \setminus \{0\}$ to 0, so we need one more generator $d$ mapping a state other than 0 to 0.
\end{proof}

We are finally in a position to prove our main theorem of this section.

\begin{theorem}[Left Ideals, Suffix-Closed Languages]
\label{thm:LeftIdeal4}
Suppose that $L\subseteq\Sig^*$ and $\kappa(L)=n$.
If $L$ is a left ideal or a suffix-closed language, then $\sig(L)\le n^{n-1}+ n-1$.
This bound is tight for $n=1$  if $|\Sig|\ge 1$, for $n=2$ if $|\Sig|\ge 3$, for $n=3$ if $|\Sig|\ge 4$, and for $n\ge 4$ if $|\Sig|\ge 5$.
Moreover, the sizes of the alphabet cannot be reduced.
\end{theorem}
\begin{proof}
If $L$ is a left ideal, then $\sig(L_n)\le n^{n-1}+n-1$ by Lemma~\ref{lem:left_upper-bound}. 
By Lemma~\ref{lem:left} the languages of Definition~\ref{def:left} meet this bound.
It is easy to verify that the size of the alphabet cannot be reduced if $n \le 3$.
For $n \ge 4$, by Theorem~\ref{thm:left-ideals_unique-maximal} only languages $L$ whose quotient automaton has transition semigroup isomorphic to $T_n$ meet the bound, and by Proposition~\ref{prop:left-ideals_maximal_generators} $T_n$ requires 5 generators.
\end{proof}

\section{Two-Sided Ideals}
\label{sec:2sided}
If a language $L$ is a right ideal, then $L=L\Sig^*$ and $L$ has exactly one final quotient, namely $\Sig^*$;
hence this also holds for two-sided ideals.
For $n\ge 3$, in a two-sided ideal every maximal chain is of length at least 3: it starts with $L$, every quotient contains $L$ and is contained in $\Sig^*$.

\subsection{Lower Bound}

We now show that the syntactic complexity of the following DFA of a two-sided ideal is 
$n^{n-2} + (n-2) 2^{n-2} +1$.

\begin{definition}[Witness: Two-Sided Ideals]\label{def:2-sided}
For $n\ge 4$, 
define the DFA  
$\cW_n =(Q_n,\Sig_\cW,\delta_\cW,0,\{n-1\}),$
where  $\Sig_\cW=\{a,b,c,d,e,f\}$, 
 $a\colon (1,\ldots,n-2)$,
$b\colon (1,2)$,
$c\colon (n-2\to 1)$,
$d\colon (n-2\to 0)$,
$e\colon Q_{n-1}\to 1$,
and $f\colon (1\to n-1)$.
For $n=4$, inputs $a$ and $b$ coincide, and we can use $\Sig_\cW=\{a,c,d,e,f\}$.
Also, let $\cW_3=(Q_3,\{a,b,c\},\delta_\cW,0,\{2\})$, where  
 $a\colon (1\to 2)(0\to 1)$,  
 $b\colon ( 1\to 0)$, 
 and $c\colon \one$, 
 and let $\cW_2=(Q_2,\{a,b\},\delta_\cW,0,\{1\})$, where  
 $a\colon (0\to 1)$, and $b\colon \one$, Let $L_n=L(\cW_n)$.
\end{definition}
The structure of the DFA of Definition~\ref{def:2-sided} is shown in Fig.~\ref{fig:2-sided} for $n\ge 4$.

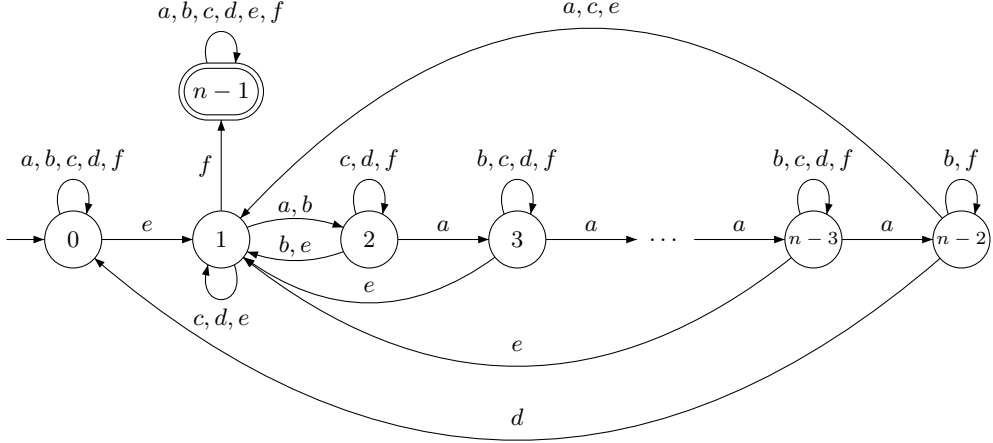
\begin{figure}[ht]
\unitlength 8pt\footnotesize
\begin{center}\begin{picture}(44,22)(0,-2)
\gasset{Nh=2.7,Nw=2.7,Nmr=1.35,ELdist=0.4,loopdiam=1.5}
\node(0)(1,8){0}\imark(0)
\node(1)(8,8){1}
\node(2)(15,8){2}
\node(3)(22,8){3}
\node[Nframe=n](3dots)(29,8){$\dots$}
\node(n-3)(36,8){\scriptsize$n-3$}
\node(n-2)(43,8){\scriptsize$n-2$}
\node[Nw=4](n-1)(8,15){$n-1$}\rmark(n-1)
\drawloop(n-1){$a,b,c,d,e,f$}
\drawedge(1,n-1){$f$}
\drawedge(0,1){$e$}
\drawloop(0){$a,b,c,d,f$}
\drawloop[loopangle=270](1){$c,d,e$}
\drawedge[curvedepth= 1,ELdist=.1](1,2){$a,b$}
\drawedge[curvedepth= 1,ELdist=-1.2](2,1){$b,e$}
\drawloop(2){$c,d,f$}
\drawedge(2,3){$a$}
\drawedge[curvedepth=3,ELdist=-1](3,1){$e$}
\drawedge(3,3dots){$a$}
\drawedge(3dots,n-3){$a$}
\drawloop(3){$b,c,d,f$}
\drawloop(n-3){$b,c,d,f$}
\drawedge(n-3,n-2){$a$}
\drawedge[curvedepth=6,ELdist=-1.2](n-3,1){$e$}
\drawedge[curvedepth=-10,ELdist=-1.2](n-2,1){$a,c,e$}
\drawedge[curvedepth=9.5,ELdist=-1.5](n-2,0){$d$}
\drawloop(n-2){$b,f$}
\end{picture}\end{center}
\caption{Quotient DFA of a two-sided ideal with $n^{n-2} + (n-2) 2^{n-2} +1$ transformations.}
\label{fig:2-sided}
\end{figure}

\begin{lemma}\label{lem:2-sided}
For $n\ge 2$, the DFA of Definition~\ref{def:left} is minimal, accepts a two-sided ideal, and its transition semigroup has size $n^{n-2} + (n-2) 2^{n-2} +1$.
\end{lemma}
\begin{proof}
For $i=1,\ldots,n-2$, state $i$ is the only non-final state that accepts $a^{n-1-i}f$; hence all these states are distinguishable. State 0 is distinguishable from these states, because it does not accept any words in $a^*f$. Hence $\cW_n$ is minimal.
The proof that $\cW_n$ is a left ideal is like that in Lemma~\ref{lem:left}.
Since $n-1$ is the only final state, $L_n$ is a right ideal. Hence it is two-sided.

For $n=3$, $\cW_3$ meets the bound $6$ with the transition semigroup consisting of the transformations $[0,1,2]$, $[1,2,2]$, $[2,2,2]$, $[0,0,2]$, $[1,1,2]$, and $[0,2,2]$. Also, for $n=2$, $\cW_2$ meets the bound $2$.

From now on we may assume that $n\ge 4$.
In $\cW_n$, the transformations induced by $a$, $b$, and $c$ restricted to $Q_n \setminus \{0,n-1\}$ generate all the transformations of the states $1,\ldots,n-2$. Together with the transformations of $d$ and $f$, they generate all $n^{n-2}$ transformations of $Q_n$ that fix $0$ and $n-1$.
For any subset $S \subseteq \{1,\ldots,n-2\}$, there is a transformation---induced by a word $w_S$, say---that maps $S$ to $n-1$ and fixes $Q_n \setminus S$. 
Then the words of the form $w_S e a^i$, for $i \in \{0,\ldots,n-3\}$, induce all transformations that map $S \cup \{n-1\}$ to $n-1$ and $Q_n \setminus (S \cup \{n-1\})$ to $i+1$. 
There are $2^{n-2}$ such transformations, and for each such transformation there are $n-2$ possibilities for $i$. Hence there are $(n-2)2^{n-2}$  transformations of this type.
There is also the constant transformation $ef\colon (Q_n \to n-1)$, which yields the total number claimed.
\end{proof}

\subsection{Upper Bound}

We consider a minimal DFA $\cD_n=(Q_n, \Sigma_\cD, \delta_\cD, 0,\{n-1\})$ of an arbitrary two-sided ideal with $n$ quotients, and let $T_n$ be the transition semigroup of $\cD_n$. 
We also deal with the witness DFA $\cW_n =(Q_n,\Sig_\cW,\delta_\cW,0,\{n-1\})$ of Definition~\ref{def:2-sided} with transition semigroup $S_n$.

\begin{lemma}\label{lem:chain3}
If $n \ge 4$ and a maximal-length chain in $\cD_n$ strictly ordered by $\prec$ has length 3, then $|T_n|\le n^{n-2}+(n-2)2^{n-2}+1$, and $T_n$ is a subsemigroup of $S_n$.
\end{lemma}
\begin{proof}
Consider an arbitrary transformation $t\in T_n$; then $(n-1)t=n-1$. 
If $0t=0$, then any state not in $\{0,n-1\}$ can possibly be mapped by $t$ to any one of the $n$ states; hence there are at most $n^{n-2}$ such transformations.

If $0t\neq 0$, then
$0\prec 0t$. Consider any state $q\not \in \{0,0t\}$; since $\cD_n$ is minimal, $q$ must be reachable from 0 by some transformation $s$, that is, $q=0s$.
If $0st \not\in \{0t,n-1\}$, then $0t\prec 0st$ by Remark~\ref{rem:left-ideals_xy2}.
But then we have the chain $0\prec 0t\prec 0st\prec n-1$ of length 4, contradicting our assumption.
Hence we must have either $0st=0t$, or $0st=n-1$.
For a fixed $0t$, a subset of the states in $Q_n\setminus\{0,n-1\}$ can be mapped to $0t$ and the remaining states in $Q_n\setminus\{0,n-1\}$ to $n-1$, thus giving $2^{n-2}$ transformations.
Since there are $n-2$ possibilities for $0t$, we obtain the second part of the bound.
Finally, all states can be mapped to $n-1$.

By Lemma~\ref{lem:2-sided} all of the above-mentioned transformations are in $S_n$.
\end{proof}

\begin{lemma}[Two-Sided Ideals, Factor-Closed Languages]\label{lem:2sided-ideals_upper-bound}
If $L$ is a two-sided ideal or a factor-closed language with $n\ge 2$ quotients, then its syntactic complexity is less than or equal to 
$n^{n-2} + (n-2) 2^{n-2} +1$.
\end{lemma}
\begin{proof}
It suffices to prove the result for two-sided ideals, since factor-closed languages are their complements.

If $n=1$, the only two-sided ideal is $\Sig^*$, its syntactic complexity is 1, and so the upper bound is 1. 
If $n=2$, each two-sided ideal is of the form $L=\Sig^*\Gamma\Sig^*$, where $\emp\subsetneq \Gamma\subseteq \Sigma$, its syntactic complexity is $2$, and so the upper bound is 2, and this agrees with Lemma~\ref{lem:2-sided}.
If $n=3$, there are eight transformations that are initially aperiodic and such that $(n-1)t = n-1$ (the property of a right-ideal transformation).
We have verified that the DFA having all eight or any seven of the eight transformations is not a two-sided ideal. Hence $6$ is an upper bound.
From now on we may assume that $n\ge 4$.

As we did for left ideals,
we show that $|T_n| \le |S_n|$, by constructing an injective function $f\colon T_n \to S_n$.

We have $q \preceq n-1$ for all $q \in Q_n$, and $n-1$ is a fixed point of every transformation in $T_n$ and $S_n$.

For a transformation $t \in T_n$, consider the following cases:
\smallskip

\noin
\hglue 15 pt
{\bf Case 1:} $t \in S_n$. \\
The proof is the same as that of Case 1 of Lemma~\ref{lem:left_upper-bound}.
\smallskip

\noin
\hglue 15 pt
{\bf Case 2:} $t \not\in S_n$, and $0 t^2 \neq 0t$. \\
Let $0t=p \prec \dots \prec p t^k = p t^{k+1}$ be the chain defined from $p$.
\smallskip

\noin
\hglue 15 pt
$\bullet$ {\bf (a):} $p t^k \neq n-1$. \\
The proof is the same as that of Case 2 of Lemma~\ref{lem:left_upper-bound}.\smallskip

\noin
\hglue 15 pt
$\bullet$ {\bf (b):} $p t^k = n-1$ and $k \ge 2$.\\
Let $f(t)=s$, where $s$ is the transformation shown in Fig.~\ref{fig:2sided-case2b} and defined by
\begin{center}
$0 s = 0, \quad p t^i s = p t^{i-1} \text{ for }1 \le i \le k-1 , \quad p s = n-1,$

$q s = q t \text{ for the other states } q\in Q_n.$
\end{center}
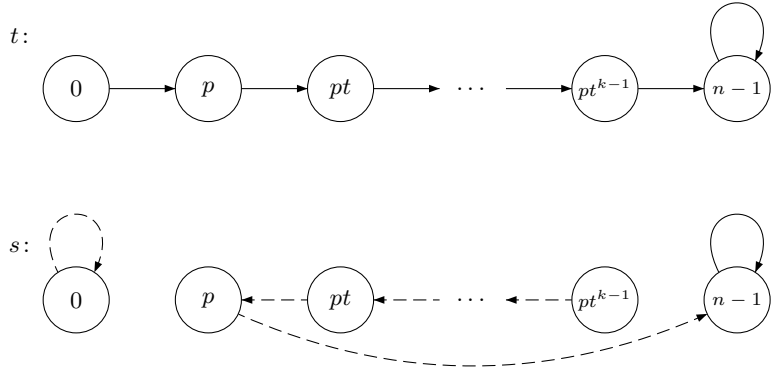
\begin{figure}[ht]
\unitlength 10pt\footnotesize
\begin{center}\begin{picture}(27,14)(0,0)
\gasset{Nh=2.5,Nw=2.5,Nmr=1.25,ELdist=0.5,loopdiam=2}
\node[Nframe=n](name)(0,12){$t\colon$}
\node(0)(2,10){0}
\node(p)(7,10){$p$}
\node(pt)(12,10){$pt$}
\node[Nframe=n](pdots)(17,10){$\dots$}
\node(pt^{k-1})(22,10){\scriptsize$pt^{k-1}$}
\node(n-1)(27,10){\scriptsize$n-1$}
\drawedge(0,p){}
\drawedge(p,pt){}
\drawedge(pt,pdots){}
\drawedge(pdots,pt^{k-1}){}
\drawedge(pt^{k-1},n-1){}
\drawloop[loopangle=90](n-1){}

\node[Nframe=n](name)(0,4){$s\colon$}
\node(0')(2,2){0}
\node(p')(7,2){$p$}
\node(pt')(12,2){$pt$}
\node[Nframe=n](pdots')(17,2){$\dots$}
\node(pt^{k-1}')(22,2){\scriptsize$pt^{k-1}$}
\node(n-1')(27,2){\scriptsize$n-1$}
\drawloop[loopangle=90,dash={.5 .25}{.25}](0'){}
\drawedge[ELdist=-1,dash={.5 .25}{.25}](pt',p'){}
\drawedge[ELdist=-1,dash={.5 .25}{.25}](pdots',pt'){}
\drawedge[ELdist=-1,dash={.5 .25}{.25}](pt^{k-1}',pdots'){}
\drawedge[curvedepth=-2.5,ELdist=-1,dash={.5 .25}{.25}](p',n-1'){}
\drawloop[loopangle=90](n-1'){}
\end{picture}\end{center}
\caption{Case~2(b) in the proof of Lemma~\ref{lem:2sided-ideals_upper-bound}.}
\label{fig:2sided-case2b}
\end{figure}
By Lemma~\ref{lem:2-sided}, $s \in S_n$.
We have $pt \succ p$, $pts = p$, and $ps = n-1$.
By Proposition~\ref{prop:chain}, $pts\succeq ps$, that is, $p \succeq n-1$, which contradicts the fact that $p \neq n-1$ (since $k \ge 2$), and $q \preceq n-1$ for all $q \in Q_n$.

Thus $s$ is not in $T_n$, and so it is different from the transformations of Case~1.

Observe that $s$ does not have a cycle with states strictly ordered by $\prec$, since no state from $\{0, p, pt, \ldots, pt^{k-1} \}$ can be in a cycle, and $t$ cannot have such a cycle. Hence $s$ is different from the transformations of Case~2(a).

In $s$, there is a unique state $q$ such that $q s = n-1$ and for which there exists a state 
$r$ such that $r \succ q$ and $r s = q$, and that this state $q$ must be $p$.
Indeed, if $q \neq p$, then $q t = q s = n-1$ by the definition of $s$.
From $r \succ q$, we have $rt\succeq qt=n-1$; hence $rs=rt=n-1$ and $rt\neq q$---a contradiction. Hence $q=p$.

By a similar argument, we show that there exists a unique state $q$ such that $q \succ p$, and $q s = p$, and that this state $q$ must be $pt$. If $q \neq pt$ then $qs = qt$. But $q \succ qt$ and $p=qt \succeq qt^2=pt$ contradicts that $p \prec pt$. Continuing in this way for $pt^2,\ldots,pt^{k-1}$ we show that there is a unique chain
$pt^{k-1} \stackrel{s }{\rightarrow} \dots \stackrel{s }{\rightarrow} pt \stackrel{s }{\rightarrow} p$.

If $t'\neq t$ is another transformation satisfying the conditions of this case, we define $s'$ like $s$.
Now suppose that $s = f(t) = f(t') = s'$.
Since we have a unique state $p$ such that $p s = n-1$ for which there exists a state $r$ such that $r \succ p$ and $r s = p$, we have $0 t = 0 t' = p$.
Also the chain of states $p,pt,pt^2,\dots,pt^{k-1}$ is unique in $s$ and $s'$ as we have shown above; so $p t^{i} = p {t'}^{i}$ for $i=1,\dots,k-1$. Since the other states are mapped by $s$ exactly as by $t$ and $t'$, we have $t = t'$.
\smallskip

\noin
\hglue 15 pt
$\bullet$ {\bf (c):} $p t = n-1$.\\
Let $P=\{0,p,n-1\}$. Since $n\ge 4$, there must be a state $r\not\in P$. 
If $p\prec r$ for all $r\not\in P$, then $n-1=pt\preceq rt$; hence $rt=n-1$ for all such $r$,
and $q t \in \{p,n-1\}$ for all $q \in Q_n$.
By Lemma~\ref{lem:2-sided}, there is a transformation in $S_n$ that maps $S \cup \{n-1\}$ to $n-1$, and $Q_n \setminus (S \cup \{n-1\})$ to $p$ for any $S \subseteq \{1,\ldots,n-2\}$. Thus $t \in S_n$---a contradiction.

In view of the above, there must exist a state $r\not\in P$ such that $p\not\preceq r$.
By Remark~\ref{rem:left-ideals_xy2}, we have $p\preceq rt$ and of course $rt\preceq n-1$. If $rt$ is $p$ or $n-1$ for all $r \not\in P$, we again have the situation described above, showing that $t \in S_n$.
Hence there must exist an $r\not\in P$ such that $p\not\preceq r$ and $p \prec r t \prec n-1$.

Also we claim that $t$ does not have a cycle. Indeed, if $p \preceq q$, then $q$ is mapped to $n-1$; if $p \not\preceq q$, then $q$ is mapped to a state $q t \succeq p$ and again $q$ cannot be in a cycle since the chain starting with $q$ ends in $n-1$.

Let $f(t)=s$, where $s$ is the transformation shown in Fig.~\ref{fig:2sided-case2c} and defined by
\begin{center}
$0 s = 0, \quad p s = r t, \quad (r t) s = p, \quad r s = 0,$

$q s = q t \text{ for the other states } q\in Q_n.$
\end{center}
\begin{figure}[ht]
\unitlength 10pt\footnotesize
\begin{center}\begin{picture}(18,16)(0,1)
\gasset{Nh=2.5,Nw=2.5,Nmr=1.25,ELdist=0.5,loopdiam=2}
\node[Nframe=n](name)(0,14){$t\colon$}
\node(0)(2,10){0}
\node(p)(10,10){$p$}
\node(n-1)(18,10){\scriptsize$n-1$}
\node(r)(6,14){$r$}
\node(rt)(14,14){$rt$}
\drawedge(0,p){}
\drawedge(p,n-1){}
\drawloop[loopangle=90](n-1){}
\drawedge(r,rt){}
\drawedge(rt,n-1){}

\node[Nframe=n](name)(0,6){$s\colon$}
\node(0')(2,2){0}
\node(p')(10,2){$p$}
\node(n-1')(18,2){\scriptsize$n-1$}
\node(r')(6,6){$r$}
\node(rt')(14,6){$rt$}
\drawloop[loopangle=90,dash={.5 .25}{.25}](0'){}
\drawedge[curvedepth=1,dash={.5 .25}{.25}](p',rt'){}
\drawedge[curvedepth=1,dash={.5 .25}{.25}](rt',p'){}
\drawedge[dash={.5 .25}{.25}](r',0'){}
\drawloop[loopangle=90](n-1'){}
\end{picture}\end{center}
\caption{Case~2(c) in the proof of Lemma~\ref{lem:2sided-ideals_upper-bound}.}
\label{fig:2sided-case2c}
\end{figure}
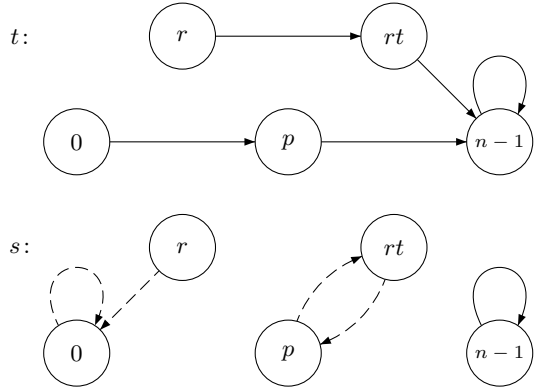

Since $s$ fixes both 0 and $n-1$, it is in $S_n$ by Lemma~\ref{lem:2-sided}. But $s$ is not in $T_n$, as we have the cycle $(p, r t)$ with $p \prec r t$. So $s$ is different from the transformations of Case~1.
Since $s$ maps a state other than $0$ to $0$, it is different from the transformations of Cases~2(a) and 2(b).

Observe that $t$ does not map any state to $0$; otherwise, if $qt = 0$ for some $q$, then $0 \prec p$ implies $q \prec 0$ by Proposition~\ref{prop:chain}, which contradicts that $0 \prec q$ from Remark~\ref{rem:left-ideals_xy2}, since $q$ is reachable from $0$ by some transformation.
Consequently, in $s$ there is the unique state $r\neq 0$ mapped to $0$. Also, as $t$ does not contain a cycle, the only cycle in $s$ must be $(p, r t)$.

If $t'\neq t$ is another transformation satisfying the conditions of this case, we define $s'$ like $s$. Now suppose that $s = f(t) = f(t') = s'$.
Because both $s$ and $s'$ have the unique non-fixed point $r$ mapped to $0$, $r = r'$. Also $s$ and $s'$ contain the unique cycle $(p, r t)$, $p \prec r t$. Thus $p = p'$, $p t = p t' = n-1$ and $r t = r t'$. It follows that $0 t = 0 t' = p$. Because $p \prec r t = r t'$, we have $(r t) t = (r t) t' = n-1$. The other states are mapped by $s$ exactly as by $t$ and $t'$, and so $t = t'$.
\smallskip

\noin
\hglue 15 pt
{\bf Case 3:} $t \not\in S_n$, $0 t = p \neq 0$, and $p t = p$. 

\noin
\hglue 15 pt
$\bullet$ {\bf (a):} $t$ has a cycle.\\
The proof is analogous to that of Case~3(a) in Lemma~\ref{lem:left_upper-bound}, 
but we need to ensure that $s$ is different from the $s$ of Cases~2(b) and 2(c).

Here there is the state $r$ such that $r \prec r s$, and $r s$ and $r s^2$ are not comparable under $\preceq$.
Consider a transformation $t'$ that fits in Case~2(b). Then in $s'$ every state $q = pt^i$ for $0 \le i \le k-1$, and $q = 0$, is mapped to a state comparable with $q$ under $\preceq$, and the other states are mapped as in $t'$. Since $t' \in T_n$ cannot contain a state $r'$ such that $r' \prec r' t$ and  $r' t$ and $r' t^2$ are not comparable under $\preceq$, it follows that $s'$ also does not contain such a state. Thus $s \neq s'$.

For a distinction from the transformations of Case~2(c) observe that $s$ does not map to 0 any state other than 0.\smallskip

\noin
\hglue 15 pt
$\bullet$ {\bf (b):} $t$ has no cycles and has a fixed point $r \not\in \{p,n-1\}$.\\
The proof is analogous to that of Case~3(b) in Lemma~\ref{lem:left_upper-bound}, 
but we need to ensure that $s$ is different from the $s$ of Cases~2(b) and 2(c).

Since $s$ maps to 0 a state other than 0, this case is distinct from Case~2(b).
Because $t$ does not have a cycle, and no state $q$ mapped to $0$ can be in a cycle in $s$, it follows that $s$ does not have a cycle. Thus $s$ is different from the transformations of Case~2(c).
\smallskip

\noin
\hglue 15 pt
$\bullet$ {\bf (c):} $t$ has neither a cycle nor a fixed point $r \not\in \{p,n-1\}$, and\ has a state $r \succ p$ mapped to $p$.\\
\noin The proof is analogous to that of Case~3(c) in Lemma~\ref{lem:left_upper-bound},
but we need to ensure that $s$ is different from the $s$ of Cases~2(b) and 2(c).

As before, since $s$ maps to 0 a state other than $0$, this case is distinct from Case~2(b).
In $s$, $0$ cannot be in a cycle, no state $q \succ p$ mapped to $0$ can be in a cycle and $p$ cannot be in a cycle as $ps = r$ and $rs = 0$. Since the other states are mapped as in $t$, $s$ does not have a cycle. Thus $s$ is different from the transformations of Case~2(c).
\smallskip

\noin
\hglue 15 pt
$\bullet$ {\bf (d):} $t$ has no cycles, no fixed point $r \not\in \{p,n-1\}$, and no state $r \succ p$ mapped to $p$, and has a state $r$ such that $p \prec r \prec n-1$ that is mapped to $n-1$.\\
Let $f(t)=s$, where $s$ is the transformation shown in Fig.~\ref{fig:2sided-case3d} and defined by
\begin{center}
$0 s = 0, \quad q s = q \text{ for states } q \text{ such that } q t = n-1, \quad p s = n-1$
$qs = qt \text{ for the other states } q\in Q_n.$
\end{center}
\begin{figure}[ht]
\unitlength 10pt\footnotesize
\begin{center}\begin{picture}(22,14)(0,0)
\gasset{Nh=2.5,Nw=2.5,Nmr=1.25,ELdist=0.5,loopdiam=2}
\node[Nframe=n](name)(0,12){$t\colon$}
\node(0)(2,10){0}
\node(p)(7,10){$p$}
\node(r)(12,10){$r$}
\node[Nframe=n](rdots)(17,10){$\dots$}
\node(n-1)(22,10){\scriptsize$n-1$}
\drawedge(0,p){}
\drawloop[loopangle=90](p){}
\drawloop[loopangle=90](n-1){}
\drawedge[curvedepth=-1.5,ELdist=-1](r,n-1){}
\drawedge(rdots,n-1){}

\node[Nframe=n](name)(0,4){$s\colon$}
\node(0')(2,2){0}
\node(p')(7,2){$p$}
\node(n-1')(22,2){\scriptsize$n-1$}
\node(r')(12,2){$r$}
\node[Nframe=n](rdots')(17,2){$\dots$}
\drawloop[loopangle=90,dash={.5 .25}{.25}](0'){}
\drawloop[loopangle=90,dash={.5 .25}{.25}](r'){}
\drawloop[loopangle=90,dash={.5 .25}{.25}](rdots'){}
\drawedge[curvedepth=-2.5,ELdist=-1,dash={.5 .25}{.25}](p',n-1'){}
\drawloop[loopangle=90](n-1'){}
\end{picture}\end{center}
\caption{Case~3(d) in the proof of Lemma~\ref{lem:2sided-ideals_upper-bound}.}
\label{fig:2sided-case3d}
\end{figure}
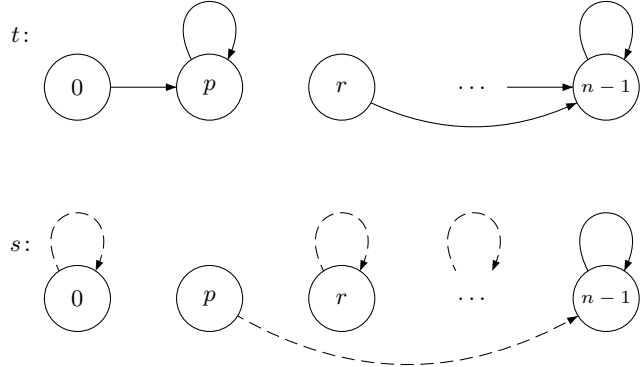

By Lemma~\ref{lem:2-sided}, $s \in S_n$. However, $s$ is not in $T_n$, as we have a fixed point $r$ such that $p \prec r \prec n-1$ and $p s = n-1$. So Proposition~\ref{prop:chain} yields $n-1 = p s \preceq r s = r$---a contradiction. Thus $s$ is different from the transformations of Case~1.

Transformation $s$ does not have any cycles, as $t$ does not have one in this case and fixed points $q$ and $p$ cannot be in a cycle. So $s$ is different from the transformations of Cases~2(a) and~3(a).
Also, since  $p$ is the unique state mapped to $n-1$ and there is no state $r \succ p$ mapped to $p$,  $s$ is different from the transformations of Case~2(b).
For a distinction from the transformations of Cases~2(c), 3(b) and 3(c), observe that $s$ does not map to 0 any state other than $0$.

If $t'\neq t$ is another transformation satisfying the conditions of this case, we define $s'$ like $s$.
Now suppose that $s = f(t) = f(t') = s'$.
Observe that $t$ does not have a fixed point other than $n-1$. So for every fixed point $q \not\in \{0,n-1\}$ of $s$ we have $q t = q t' = n-1$. Also, since $p$ is the unique state mapped to $n-1$ in $s$, $0 t = 0 t' = p$ and $p t = p t' = p$. The other states are mapped by $s$ as by $t$ and $t'$; so $t = t'$.
\smallskip

\hglue 5pt $\bullet$ {\bf All cases are covered:} \\ 
See the appendix for a list of the cases.
We need to ensure that any transformation $t$ fits in at least one case. It is clear that $t$ fits in Case~1 or 2 or 3. Any transformation from Case~2 fits in Case~2(a) or 2(b) or 2(c).
For Case~3, it is sufficient to show that if (i) $t \not\in S_n$ does not contain a fixed point $r\not\in \{p,n-1\}$, and (ii) there is no state $r$, $p \prec r \prec n-1$, mapped to $p$ or $n-1$, then $t$ has a cycle.

If there is no state $r$ such that $p \prec r \prec n-1$, then $q t \in \{p,n-1\}$ for all $q \in Q_n$, since $q t \succeq p$.
By the proof of Lemma~\ref{lem:2-sided} in $S_n$ for any $S \subseteq Q_n \setminus \{n-1\}$ there are all transformations that map $S \cup \{n-1\}$ to $n-1$, and the other states $Q_n \setminus (S \cup \{n-1\})$ to any state from $Q_n$; thus $t \in S_n$---a contradiction.

So let $t$ be a transformation that fits in Case~3 and satisfies~(i) and~(ii), and let $r$ be some state such that $p\prec r \prec n-1$.
Consider the sequence $r, rt, rt^2, \ldots$.
By Remark~\ref{rem:left-ideals_xy2}, $p \preceq rt^i$ for all $i \ge 0$.
If $rt^k \in \{p,n-1\}$ for some $k \ge 1$, let $k$ be the smallest such number, then $rt^{k-1} \notin \{p,n-1\}$; we have $p \prec rt^{k-1} \prec n-1$ and $(rt^{k-1}) t \in \{p,n-1\}$, contradicting (ii).

Since $p$ and $n-1$ are the only fixed points by (i), we have $rt^i \neq rt^{i-1}$. Since there are finitely many states, $rt^i = rt^j$ for some $i$ and $j$ such that $0 \le i < j-1$, and so the states $rt^i,rt^{i+1}\ldots,rt^j=rt^i$ form a cycle.
\end{proof}

\begin{theorem}\label{thm:2sided-ideals_unique-maximal}
If $T_n$ has size $n^{n-2} + (n-2) 2^{n-2} +1$, then $T_n=S_n$.
\end{theorem}
\begin{proof}
The proof is very similar to that of Theorem~\ref{thm:left-ideals_unique-maximal}.

Consider a maximal-length chain of states strictly ordered by $\prec$ in $\cD_n$. If its length is 3, then by Lemma~\ref{lem:chain3} $T_n$ is a subsemigroup of $S_n$. Thus only $T_n = S_n$ reaches the bound. 

If there is a chain of length 4, 
then there are states $p$ and $r$ such that $0 \prec p \prec r \prec n-1$. Let $f$ be the injective function from Lemma~\ref{lem:2sided-ideals_upper-bound}.
Consider the transformation $u$ that maps $Q_n \setminus \{n-1\}$ to $p$ and fixes $n-1$.
Let $s$ be defined from $u$ in Case~3(c) of the proof of Lemma~\ref{lem:2sided-ideals_upper-bound}. The rest of the proof follows the proof of Theorem~\ref{thm:left-ideals_unique-maximal} with Case~3(d) of Lemma~\ref{lem:2sided-ideals_upper-bound} added.
\end{proof}

\begin{proposition}\label{prop:2sided-ideals_maximal_generators}
For $n \ge 4$, the minimal number of generators of the transition semigroup $T_n$ is 6.
\end{proposition}
\begin{proof}
From Proposition~\ref{prop:left-ideals_maximal_generators} we know that the transformations in $T_n$ restricted to $Q_n \setminus \{n-1\}$ require 5 generators.
These generators in $T_n$ do not map any state from $Q_n \setminus \{n-1\}$ to $n-1$, and must fix $n-1$.
Hence, we need one more generator that map a state from $Q_n \setminus \{n-1\}$ to $n-1$.
\end{proof}

We are now in a position to prove our main theorem of this section.

\begin{theorem}[Two-Sided Ideals, Factor-Closed Languages]
\label{thm:2sided}
Suppose that $L\subseteq\Sig^*$ and $\kappa(L)=n>1$.
If $L$ is a two-sided ideal or a factor-closed language, then  $\sig(L)\le n^{n-2} + (n-2) 2^{n-2} +1$.
This bound is tight for $n=2$ if $|\Sig|\ge 2$, for $n=3$ if $|\Sig|\ge 3$, for $n\ge 4$ if $|\Sig|\ge 5$, and for $n\ge 5$ if $|\Sig|\ge 6$.
Moreover, the sizes of the alphabet cannot be reduced.
\end{theorem}
\begin{proof}
This follows from Lemmas~\ref{lem:2-sided} and~\ref{lem:2sided-ideals_upper-bound}.
It is easy to verify that the size of the alphabet cannot be reduced if $n \le 4$.
For $n \ge 5$, by Theorem~\ref{thm:2sided-ideals_unique-maximal} only languages $L$ whose quotient automaton has transition semigroup isomorphic to $T_n$ meet the bound, and by Proposition~\ref{prop:2sided-ideals_maximal_generators}, semigroup $T_n$ requires 6 generators.
\end{proof}

\section{Conclusions}\label{sec:conclusions}

We have found tight upper bounds on the syntactic complexity of right, left, and two-sided ideals.
Despite the fact that the Myhill congruence has left-right symmetry, there are significant differences between left and right ideals. 
We have shown that in each of the three cases the maximal transition semigroup is unique.

In our proof for left and two-sided ideals we exhibited an injective function from the transition semigroup of a minimal DFA of an arbitrary left, right, two-sided ideal language to the transition semigroup of the witness DFA attaining the upper bound for these languages. This approach is generally applicable for other subclasses of regular languages.
For example, in~\cite{BrSz15} we have used this method to establish the upper bound for suffix-free languages.

\section*{Acknowledgements}
          
This work was supported by the Natural Sciences and Engineering Research Council of Canada (NSERC)
grant No.~OGP000087, the National Science Centre, Poland under project number 2013/09/N/ST6/01194,
an NSERC Postgraduate Scholarship, and a Graduate Award from the Department of Computer Science, University of Toronto.
It was completed during the internship of Marek Szyku{\l}a at the University of Waterloo, which was co-financed by the European Union under the European Social Fund's project ``International computer science and applied mathematics for business study programme at the University of Wroc{\l}aw''.\\

\bibliographystyle{spmpsci}
\providecommand{\noopsort}[1]{}

\newpage
\section*{Appendix}

\subsection*{List of the cases in the proof of Lemma~\ref{lem:left_upper-bound}.}

{\setstretch{1.1}
\begin{enumerate}[leftmargin=*,widest=\textbf{Case~1}]
\item[\textbf{Case~1}:] $t \in S_n$.
\medskip
\item[\textbf{Case~2}:] $t \not\in S_n$ and $0t^2 \neq 0t$.
\medskip
\item[\textbf{Case~3}:] $t \not\in S_n$ and $0t^2 = 0t$.
\begin{enumerate}[leftmargin=*,widest=\textbf{Case~3}]
\item[\textbf{Case~3(a)}:] $t$ has a cycle.
\item[\textbf{Case~3(b)}:] $t$ has no cycles and has a fixed point $r\neq p$.
\item[\textbf{Case~3(c)}:] $t$ has no cycles, has no fixed point $r\neq p$, and there is a state $r$ such that $p\prec r$ with $rt=p$.
\end{enumerate}
\end{enumerate}}

\begin{figure}[htb]
\unitlength 6.7pt\scriptsize
\gasset{Nh=2.5,Nw=2.5,Nmr=1.25,ELdist=0.5,loopdiam=1.5}
\begin{center}\begin{picture}(29,12)(0,-1)
\node[Nframe=n](name)(14,8){Case 1: $\textcolor{red}{s} = t$.}
\end{picture}\begin{picture}(29,12)(-2,-1)
\node[Nframe=n](name)(14,8){Case 2:}
\node(0)(2,2){0}
\node(p)(8,2){$p$}
\node(pt)(14,2){$pt$}
\node[Nframe=n](pdots)(20,2){$\dots$}
\node(pt^k)(26,2){$pt^k$}
\drawedge(0,p){}
\drawedge(p,pt){}
\drawedge(pt,pdots){}
\drawedge(pdots,pt^k){}
\drawloop[loopangle=90](pt^k){}
\drawloop[linecolor=red,dash={.5 .25}{.25},loopangle=90](0){}
\drawedge[linecolor=red,dash={.5 .25}{.25},curvedepth=2](pt^k,p){}
\end{picture}\end{center}
\begin{center}\begin{picture}(29,9)(0,-1)
\node[Nframe=n](name)(14,8){Case 3(a):}
\node(0)(2,2){0}
\node(p)(8,2){$p$}
\node(r)(14,2){$r$}
\node(rt)(16.5,5){$rt$}
\node[Nframe=n](rdots)(19,2){$\dots$}
\drawedge(0,p){}
\drawloop[loopangle=90](p){}
\drawedge[curvedepth=1](r,rt){}
\drawedge[curvedepth=1](rt,rdots){}
\drawedge[curvedepth=1](rdots,r){}
\drawloop[linecolor=red,dash={.5 .25}{.25},loopangle=90](0){}
\drawedge[linecolor=red,dash={.5 .25}{.25}](p,r){}
\end{picture}\begin{picture}(29,9)(-2,-1)
\node[Nframe=n](name)(14,8){Case 3(b):}
\node(0)(2,2){0}
\node(p)(10,2){$p$}
\node(r)(18,2){$r$}
\node[Nframe=n](rdots)(26,2){$\dots$}
\drawedge(0,p){}
\drawloop[loopangle=90](p){}
\drawloop[loopangle=90](r){}
\drawloop[loopangle=90](rdots){}
\drawloop[linecolor=red,dash={.5 .25}{.25},loopangle=90](0){}
\drawedge[linecolor=red,dash={.5 .25}{.25},loopangle=90,curvedepth=2,ELpos=30](r,0){}
\drawedge[linecolor=red,dash={.5 .25}{.25},loopangle=90,curvedepth=3,ELpos=30,exo=-1](rdots,0){}
\end{picture}\end{center}
\begin{center}\begin{picture}(29,9)(0,-1)
\node[Nframe=n](name)(14,8){Case 3(c):}
\node(0)(2,2){0}
\node(p)(10,2){$p$}
\node(r)(18,2){$r$}
\node[Nframe=n](rdots)(26,2){$\dots$}
\drawedge(0,p){}
\drawloop[loopangle=90](p){}
\drawedge(r,p){}
\drawedge[curvedepth=-3](rdots,p){}
\drawloop[linecolor=red,dash={.5 .25}{.25},loopangle=90](0){}
\drawedge[linecolor=red,dash={.5 .25}{.25},curvedepth=2,ELpos=30](r,0){}
\drawedge[linecolor=red,dash={.5 .25}{.25},curvedepth=3,ELpos=30,exo=-1](rdots,0){}
\end{picture}\begin{picture}(29,9)(0,-1)
\end{picture}\end{center}
\caption{Map of the cases in the proof of Lemma~\ref{lem:left_upper-bound}. The transitions of $t$ are represented by solid lines, and the transitions of $s$ by dashed red lines.}
\end{figure}
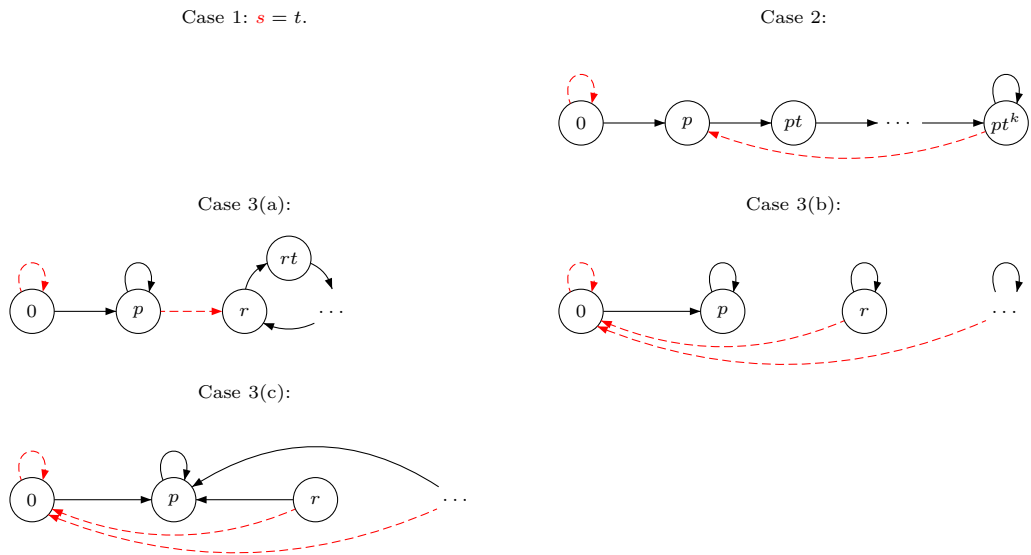


\newpage
\subsection*{List of the cases in the proof of Lemma~\ref{lem:2sided-ideals_upper-bound}.}

{\setstretch{1.1}
\begin{enumerate}[leftmargin=*,widest=\textbf{Case~1}]
\item[\textbf{Case~1}:] $t \in S_n$.
\medskip
\item[\textbf{Case~2}:] $t \not\in S_n$ and $0 t^2 \neq 0t$.
\begin{enumerate}[leftmargin=*,widest=\textbf{Case~2}]
\item[\textbf{Case~2(a)}:] $t \not\in S_n$ and $0 t^2 \neq 0t$.
\item[\textbf{Case~2(b)}:] $p t^k = n-1$ and $k \ge 2$.
\item[\textbf{Case~2(c)}:] $p t = n-1$.
\end{enumerate}
\medskip
\item[\textbf{Case~3}:] $t \not\in S_n$, $0 t = p \neq 0$ and $p t = p$.
\begin{enumerate}[leftmargin=*,widest=\textbf{Case~3}]
\item[\textbf{Case~3(a)}:] $t$ has a cycle.
\item[\textbf{Case~3(b)}:] $t$ has no cycles and has a fixed point $r \not\in \{p,n-1\}$.
\item[\textbf{Case~3(c)}:] $t$ has neither a cycle nor a fixed point $r \not\in \{p,n-1\}$, and has a state $r \succ p$ mapped to $p$.
\end{enumerate}
\end{enumerate}}

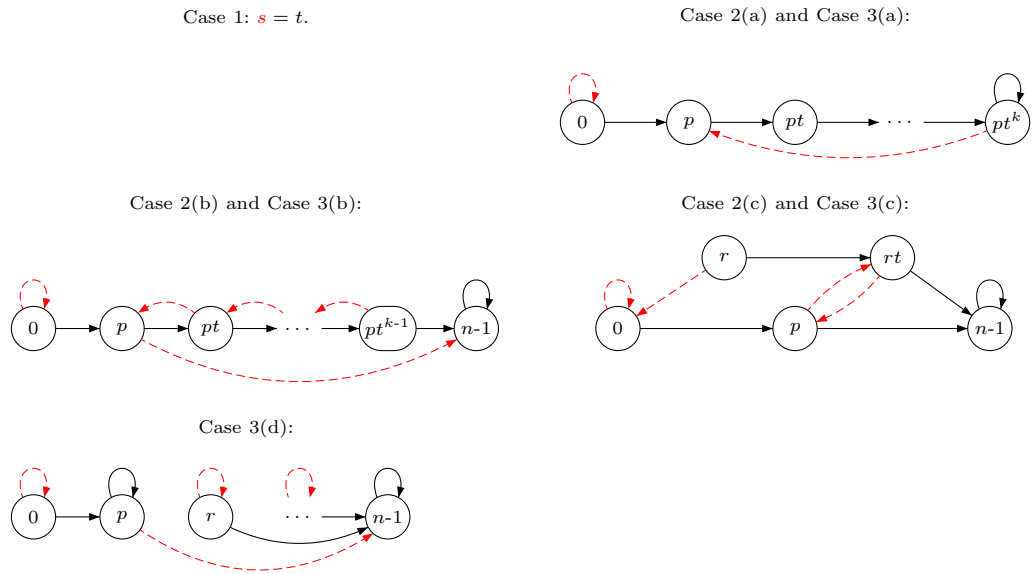
\begin{figure}[htb]
\unitlength 6.7pt\scriptsize
\gasset{Nh=2.5,Nw=2.5,Nmr=1.25,ELdist=0.5,loopdiam=1.5}
\begin{center}\begin{picture}(29,12)(0,-1)
\node[Nframe=n](name)(14,8){Case 1: $\textcolor{red}{s} = t$.}
\end{picture}\begin{picture}(29,12)(-2,-1)
\node[Nframe=n](name)(14,8){Case 2(a) and Case 3(a):}
\node(0)(2,2){0}
\node(p)(8,2){$p$}
\node(pt)(14,2){$pt$}
\node[Nframe=n](pdots)(20,2){$\dots$}
\node(pt^k)(26,2){$pt^k$}
\drawedge(0,p){}
\drawedge(p,pt){}
\drawedge(pt,pdots){}
\drawedge(pdots,pt^k){}
\drawloop[loopangle=90](pt^k){}
\drawloop[linecolor=red,dash={.5 .25}{.25},loopangle=90](0){}
\drawedge[linecolor=red,dash={.5 .25}{.25},curvedepth=2](pt^k,p){}
\end{picture}\end{center}
\begin{center}\begin{picture}(29,10)(0,-1)
\node[Nframe=n](name)(14,9){Case 2(b) and Case 3(b):}
\node(0)(2,2){0}
\node(p)(7,2){$p$}
\node(pt)(12,2){$pt$}
\node[Nframe=n](pdots)(17,2){$\dots$}
\node[Nw=3.2](pt^{k-1})(22,2){\scriptsize$pt^{k\text{-}1}$}
\node(n-1)(27,2){\scriptsize$n\text{-}1$}
\drawedge(0,p){}
\drawedge(p,pt){}
\drawedge(pt,pdots){}
\drawedge(pdots,pt^{k-1}){}
\drawedge(pt^{k-1},n-1){}
\drawloop[loopangle=90](n-1){}
\drawloop[linecolor=red,dash={.5 .25}{.25},loopangle=90](0){}
\drawedge[linecolor=red,dash={.5 .25}{.25},curvedepth=-1.5,ELdist=-1](pt,p){}
\drawedge[linecolor=red,dash={.5 .25}{.25},curvedepth=-1.5,ELdist=-1](pdots,pt){}
\drawedge[linecolor=red,dash={.5 .25}{.25},curvedepth=-1.5,ELdist=-1](pt^{k-1},pdots){}
\drawedge[linecolor=red,dash={.5 .25}{.25},curvedepth=-3,ELdist=-1](p,n-1){}
\end{picture}\begin{picture}(29,10)(-2,-1)
\node[Nframe=n](name)(14,9){Case 2(c) and Case 3(c):}
\node(0)(4,2){0}
\node(p)(14,2){$p$}
\node(n-1)(25,2){\scriptsize$n\text{-}1$}
\node(r)(10,6){$r$}
\node(rt)(19.5,6){$rt$}
\drawedge(0,p){}
\drawedge(p,n-1){}
\drawloop[loopangle=90](n-1){}
\drawedge(r,rt){}
\drawedge(rt,n-1){}
\drawloop[linecolor=red,dash={.5 .25}{.25},loopangle=90](0){}
\drawedge[linecolor=red,dash={.5 .25}{.25},curvedepth=.7](p,rt){}
\drawedge[linecolor=red,dash={.5 .25}{.25},curvedepth=.7](rt,p){}
\drawedge[linecolor=red,dash={.5 .25}{.25}](r,0){}
\end{picture}\end{center}
\begin{center}\begin{picture}(29,9)(0,-1)
\node[Nframe=n](name)(14,7){Case 3(d):}
\node(0)(2,2){0}
\node(p)(7,2){$p$}
\node(r)(12,2){$r$}
\node[Nframe=n](rdots)(17,2){$\dots$}
\node(n-1)(22,2){\scriptsize$n\text{-}1$}
\drawedge(0,p){}
\drawloop[loopangle=90](p){}
\drawloop[loopangle=90](n-1){}
\drawedge[curvedepth=-1.5,ELdist=-1](r,n-1){}
\drawedge(rdots,n-1){}
\drawloop[linecolor=red,dash={.5 .25}{.25},loopangle=90](0){}
\drawloop[linecolor=red,dash={.5 .25}{.25},loopangle=90](r){}
\drawloop[linecolor=red,dash={.5 .25}{.25},loopangle=90](rdots){}
\drawedge[linecolor=red,dash={.5 .25}{.25},curvedepth=-3,ELdist=-1,exo=.5](p,n-1){}
\end{picture}\begin{picture}(29,9)(-2,-1)
\end{picture}\end{center}
\caption{Map of the cases in the proof of Lemma~\ref{lem:2sided-ideals_upper-bound}. The transitions of $t$ are represented by solid lines, and the transitions of $s$ by dashed red lines.}
\end{figure}
\end{document}